%% file: main.tex
\documentclass[review, 11pt, a4paper, table]{elsarticle}

\usepackage{amsmath}
\usepackage{amsfonts}

\usepackage{amsthm}

\newtheorem{corollary}{Corollary}

\usepackage{bm}
\usepackage{graphicx,psfrag,epsf}
\usepackage{natbib}
\setcitestyle{authoryear,open={(},close={)}}
\usepackage[hidelinks]{hyperref}
\usepackage{enumerate}
\usepackage{doi}
\usepackage[bottom,hang,flushmargin]{footmisc}
\usepackage{booktabs}
\usepackage{url}
\usepackage[left=0.8in, right=0.8in, bottom = 1in, top = 1in]{geometry}
\usepackage{setspace}
\usepackage{tikz}
\usepackage{makecell}
\usetikzlibrary{positioning,arrows.meta,backgrounds}

\usepackage{caption}
\definecolor{thiefblue}{RGB}{31,99,166}
\usepackage{array}
\usepackage{float}
\usepackage{multirow}
\usepackage{enumitem}

\newif\ifblind
\blindfalse

\ifblind
  \newcommand{\pkgreco}{\texttt{[software A]}}
  \newcommand{\pkgcomb}{\texttt{[software B]}}
  \newcommand{\blindnote}{\footnote{Software names and the corresponding
    references are omitted for blind review.}}
  
\else
  \newcommand{\pkgreco}{\texttt{FoReco} \citep{Girolimetto2026-hv}}
  \newcommand{\pkgcomb}{\texttt{FoCo2} \citep{Girolimetto2025-sf}}
  \newcommand{\blindnote}{}
  
\fi

\makeatletter
\def\ps@pprintTitle{}
\makeatother

\begin{document}

\begin{frontmatter}

\title{{\bfseries Forecast combination under temporal constraints: An application to day-ahead electricity prices}}

\ifblind
\else
  \author{\vspace*{-1em}Daniele Girolimetto\corref{mycorrespondingauthor}}
  \cortext[mycorrespondingauthor]{Corresponding author}
  \ead{daniele.girolimetto@unipd.it}

  \address{Department of Statistical Sciences, University of Padua, Padova 35121, Italy \vspace*{-3.5em}}
\fi

\begin{abstract}
Day-ahead electricity prices are forecast by many competing models and at several temporal granularities, from hourly prices to block and baseload products. The resulting forecasts suffer from two distinct problems: forecasts produced at different granularities are incoherent, as the aggregates do not match the averages of their components, and no single model is the most accurate in every market, period and level. We address both problems in one step by applying the coherent forecast combination framework, recently proposed for linearly constrained multiple time series, to temporal hierarchies. The resulting temporal coherent combination pools the forecasts that the competing experts produce at all the levels of the hierarchy into a single forecast that satisfies the aggregation constraints and has minimum error variance among the linear unbiased combinations. Since the optimal solution depends on a high-dimensional error covariance matrix, we compare different estimators obtained by crossing correlation structures with linear and nonlinear shrinkage approaches, and we contrast the simultaneous solution with sequential strategies. Using the base forecasts of four model classes for the German and Spanish day-ahead markets, the combined forecasts significantly outperform the base forecasts at every temporal level and are never significantly less accurate than the best reconciled expert, often improving upon it significantly.
\end{abstract}

\end{frontmatter}
\clearpage

\section{Introduction} \label{sec:intro}

Forecasting day-ahead electricity prices is one of the most studied problems in energy analytics. Two decades of research have produced a remarkable variety of modeling approaches, from parsimonious autoregressions with exogenous variables to neural networks, gradient-boosted trees and, most recently, pretrained foundation models \citep{Weron2014-gs, Lago2021-dp, Lipiecki2026-nm}. However, the ranking of competing models changes with the market, the period and the evaluation criterion, and the forecaster who commits to a single model bets on the stability of a ranking that is known to be unstable \citep{Lago2021-dp, Petropoulos2022-ji}.

Day-ahead markets add a structural dimension to this model uncertainty. In the European Single Day-Ahead Coupling (SDAC), which includes both the German and the Spanish markets, a single auction is cleared once a day, around noon, and simultaneously determines the prices of all the hourly products of the following day.\footnote{Since 1 October 2025 the SDAC has moved to a 15-minute market time unit, which deepens the temporal hierarchy considered here ($m = 96$) without changing its structure. The data used in this paper refer to the hourly regime.} Trading and hedging, however, also involve block products spanning several consecutive hours and the baseload contract covering the whole day, such as the forward and over-the-counter contracts that settle on the arithmetic average of the corresponding hourly day-ahead prices. The same quantity is thus observed, and must be predicted, at several temporal frequencies linked by exact aggregation constraints. Models estimated separately at each frequency ignore these constraints and produce incoherent forecasts: the average of the predicted hourly prices differs from the predicted baseload price, and decisions taken at different granularities are consequently misaligned \citep{Athanasopoulos2017-zh, Kourentzes2016-zv}.

Temporal hierarchies \citep{Athanasopoulos2017-zh} turn this difficulty into an opportunity. All the frequencies obtained by non-overlapping aggregation are collected into a single hierarchy, forecast independently, and reconciled through a least-squares adjustment \citep{Stone1942-pn, Byron1978-zb, Byron1979-pv}, possibly optimized by trace minimization as in \citet{Wickramasuriya2019-hs}. Reconciled forecasts are coherent by construction and, since each frequency emphasizes different features of the series, they are typically more accurate and more robust to model misspecification than the original ones \citep{Jeon2019-xb, Spiliotis2019-ll, Nystrup2020-ey, Nystrup2021-px}. After successful applications to electric load and renewable generation \citep{Di_Fonzo2023-ae, Panagiotelis2023-se,  Abolghasemi2025-fp, Bisaglia2025-gg, Calderon-Gonzalez2026-fi}, \citet{Lipiecki2026-nm} have recently shown that temporal hierarchy forecasting (THieF) also pays in electricity price forecasting, improving four model classes of very different complexity at every temporal level in the German and Spanish markets. For a comprehensive review of forecast reconciliation, see \citet{Athanasopoulos2024-cn}.

Despite their empirical success, current reconciliation approaches treat each forecasting model independently, exploiting the information available across temporal aggregation levels while leaving unresolved the choice among competing reconciled forecasts. 
The forecast combination literature, starting with the seminal contribution of \citet{Bates1969} and comprehensively reviewed by \citet{Timmermann2006} and \citet{Wang2023-bq}, consistently shows that combining forecasts is generally preferable to selecting a single model. This perspective is closely related to forecast reconciliation, which can itself be interpreted as a forecast combination procedure acting across the different levels of a hierarchy \citep{Hollyman2021-cy, Di_Fonzo2024-ym}. When multiple models produce forecasts for an entire temporal hierarchy, information is available along two distinct dimensions: across temporal aggregation levels and across forecasting models. Exploiting only one of these dimensions at a time is therefore unlikely to make full use of the available information

In this paper we treat the two directions jointly. We apply the coherent forecast combination framework of \citet{Girolimetto2024-qe}, developed for multiple time series subject to general linear constraints, to temporal hierarchies, obtaining the combination of the forecasts of $Q\geq 1$ experts on all the levels of a temporal hierarchy that is coherent with the aggregation constraints and has minimum error variance among the linear unbiased ones. We use the term expert broadly to denote any source of base forecasts for one or more components of the temporal hierarchy, such as a statistical model, a machine learning algorithm, a forecasting service, an institutional projection, or an individual human forecaster. Since temporal aggregation constraints are a special case of the linear constraints considered there, the solution follows directly from their results by specializing the constraint matrix (\autoref{thm:occ}): it admits a closed-form expression, in a structural and in a zero-constrained form, and recovers the single-expert temporal reconciliation of \citet{Athanasopoulos2017-zh} as the special case ($Q = 1$). Its practical implementation depends on estimating the covariance matrix of the stacked base forecast errors, whose dimension increases with both the size of the hierarchy and the number of experts. To address this challenge, we compare different estimation strategies that includes four increasingly sparse correlation structures with linear and nonlinear shrinkage approaches \citep{ledoit2004, Schafer2005-jh, ledoit2020analytical, ledoit2022power}.

The paper makes three contributions. First, we integrate forecast combination and temporal reconciliation into a unified framework and provide the first evidence on its performance, for a forecasting problem, day-ahead electricity prices, in which both model uncertainty and multiple temporal granularities are of direct practical relevance. Second, we show how the estimation of the error covariance matrix drives the accuracy of the combination in the temporal setting, comparing correlation structures and linear and nonlinear shrinkage techniques, and finding that the former matters far more than the latter. Third, we compare the simultaneous solution with sequential strategies that first combine and then reconcile the forecasts.

For the empirical analysis we consider the forecasting experiment of \citet{Lipiecki2026-nm} for the German EPEX-DE and Spanish OMIE electricity markets over the 2021–2024 test period using the \textsf{R} packages \pkgreco\blindnote{} and \pkgcomb{} to reconcile and combine the base forecasts obtained by four different experts. Temporal coherent combination is at least as accurate as the best reconciled individual model at every aggregation level in both markets, and significantly more accurate than all the other reconciled experts, reducing hourly MAE by 4\% to 24\% relative to the base forecasts, depending on the expert, with gains becoming larger at coarser temporal aggregations. The robustness analysis further provides clear practical guidance. The choice of correlation structure proves to be the key modelling decision: preserving each expert’s error dependence across temporal levels is substantially more important than the particular linear or nonlinear shrinkage method adopted.

The remainder of the paper is organized as follows. Section~\ref{sec:cfc} sets up temporal hierarchies and the multiple base forecasts, derives the optimal temporal coherent combination and discusses the estimation of the error covariance matrix. Section~\ref{sec:exp} describes the data and the experimental design and presents the main results. Section~\ref{sec:robust} collects the robustness analysis and Section~\ref{sec:conclusion} concludes.

\section{Coherent forecast combination with temporal hierarchies} \label{sec:cfc}
\subsection{Temporal hierarchies} \label{sec:th}

A temporal hierarchy \citep{Athanasopoulos2017-zh} organizes a single time series, observed at its highest available frequency, into a collection of series obtained by non-overlapping temporal aggregation at coarser frequencies. Let $\{y_t\}_{t = 1,\dots,T}$ be the high-frequency series, with $m$ observations per cycle (the seasonal period), and let $\mathcal{K} = \{k_p,\dots,k_1\}$ be the set of $p$ factors of $m$ arranged in descending order, with $k_p = m$ and $k_1 = 1$. Each factor $k \in \mathcal{K}$ identifies a temporal level: $k_1 = 1$ is the high-frequency level, coinciding with the original series, while $k_p = m$ is the most aggregated, lowest-frequency level, with a single observation per cycle. Throughout, $k$ denotes a generic element of $\mathcal{K}$, and each temporal level is identified by the value of its aggregation factor, whereas the subscript in $k_j$ only refers to the position of the factor in $\mathcal{K}$.

For a factor $k \in \mathcal{K}$, the non-overlapping aggregated series is
\begin{equation}
x^{[k]}_{j} = \sum_{i = 1}^{k} w^{[k]}_{i} y_{t^\ast + (j-1)k + i - 1},
\qquad j = 1, \dots, \left\lfloor (T - t^\ast + 1)/k \right\rfloor,
\label{eq:tagg}
\end{equation}
a series with seasonal period $M_k = m/k$, while at the high-frequency level the aggregation is the identity, $x^{[1]}_{j} \equiv y_{t^\ast + j - 1}$. The weights $w^{[k]}_i$ define the nature of the aggregation: $w^{[k]}_i = 1$ corresponds to temporal sums, appropriate for flow variables, whereas $w^{[k]}_i = 1/k$ corresponds to temporal averages, the relevant choice for day-ahead electricity prices, where block and baseload products are arithmetic means of the underlying hourly prices. For notational simplicity, \eqref{eq:tagg} is written under a right-endpoint alignment, so that $t^\ast$ captures any incomplete cycle at the left boundary, so that the offset $t^\ast = T - \lfloor T/m \rfloor m + 1$ discards the first incomplete cycle and only complete high-frequency cycles enter the aggregation. When $T$ is a multiple of $m$, $t^\ast = 1$. \autoref{fig:th24} displays the resulting hierarchy for the daily cycle of $m = 24$ hourly prices considered in our application, where the factors $\mathcal{K} = \{24, 12, 8, 6, 4, 3, 2, 1\}$ define $p = 8$ temporal levels, from the hourly up to the daily series.

\begin{figure}[!tb]
\centering
\resizebox{\linewidth}{!}{%
\begin{tikzpicture}[x=0.62cm, y=0.62cm]
\def\m{24}
\foreach \k/\r in {24/0, 12/1, 8/2, 6/3, 4/4, 3/5, 2/6, 1/7}{
  \pgfmathtruncatemacro{\Mk}{\m/\k}
  \pgfmathtruncatemacro{\shade}{10 + (7-\r)*5}
  \node[anchor=east, font=\small, text=black!75] at (-0.5, {-1.15*\r + 0.42}) {$k = \k$};
  \node[anchor=west, font=\small, text=black!75] at (\m + 0.5, {-1.15*\r + 0.42}) {$M_{\k} = \Mk$};
  \foreach \j in {1,...,\Mk}{
    \pgfmathsetmacro{\xa}{(\j-1)*\k + 0.07}
    \pgfmathsetmacro{\xb}{\j*\k - 0.07}
    \draw[draw=thiefblue!70!black, line width=0.4pt, rounded corners=1.6pt, fill=thiefblue!\shade]
      (\xa, {-1.15*\r}) rectangle (\xb, {-1.15*\r + 0.84});
    \ifnum\k>2
      \node[font=\scriptsize] at ({(\xa+\xb)/2}, {-1.15*\r + 0.42}) {$x^{[\k]}_{\j}$};
    \else
      \ifnum\k=2
        \node[font=\tiny] at ({(\xa+\xb)/2}, {-1.15*\r + 0.42}) {$x^{[2]}_{\j}$};
      \else
        \node[font=\tiny] at ({(\xa+\xb)/2}, {-1.15*\r + 0.42}) {$y_{\j}$};
      \fi
    \fi
  }
}
\end{tikzpicture}}
\vskip4pt
\caption{Temporal hierarchy for the daily cycle of $m = 24$ hourly prices, with
$\mathcal{K} = \{24, 12, 8, 6, 4, 3, 2, 1\}$. Each level $k$ partitions the day into
$M_k = 24/k$ non-overlapping blocks, and each block is the average of the $k$ hourly
values it spans, as in~\eqref{eq:taggcycle}. Rows are labelled by the aggregation factor $k$ (left) and by the number of level-$k$ observations per day $M_k$ (right).}
\label{fig:th24}
\end{figure}
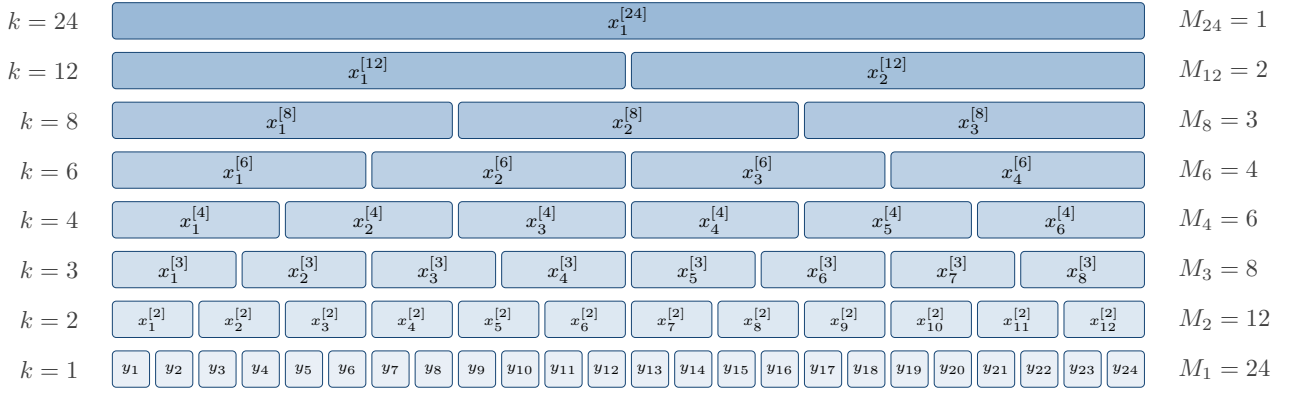

It is useful to organize the observations cycle by cycle. The sample includes $N = \lfloor T/m \rfloor$ complete cycles, and each cycle contains exactly $M_k$ level-$k$ observations, so the index decomposes as $j = (\tau - 1)M_k + s$, with $\tau = 1,\dots,N$ indexing the cycle and $s = 1,\dots,M_k$ the within-cycle position. Substituting into~\eqref{eq:tagg} and using $M_k k = m$, the level-$k$ observations of cycle $\tau$ read
\begin{equation}
x^{[k]}_{(\tau,s)} \;\equiv\; x^{[k]}_{(\tau-1)M_k + s} \;=\; \sum_{i = 1}^{k} w^{[k]}_{i} y_{t^\ast + (\tau-1)m + (s-1)k + i - 1},
\qquad s = 1,\dots,M_k.
\label{eq:taggcycle}
\end{equation}
Collecting them into $\bm{x}^{[k]}_{\tau} = \begin{bmatrix}
x^{[k]}_{(\tau,1)} \quad \dots \quad x^{[k]}_{(\tau,M_k)}
\end{bmatrix}^{\top} \in \mathbb{R}^{M_k}$, and writing the high-frequency level as $\bm{y}_{\tau} = \bm{x}^{[1]}_{\tau} \in \mathbb{R}^{m}$, the low-frequency (temporally aggregated) levels of the cycle stack into
\begin{equation}
\bm{x}_{\tau} =
\begin{bmatrix}
x^{[k_p]}_{\tau}\\
\bm{x}^{[k_{p-1}]}_{\tau}\\
\vdots\\
\bm{x}^{[k_2]}_{\tau}
\end{bmatrix} \in \mathbb{R}^{n_x},
\qquad n_x = \sum_{k \in \mathcal{K}, k > 1} M_k,
\label{eq:xvec}
\end{equation}
and the temporal hierarchy of cycle $\tau$ is summarized by the coherent vector $\bm{z}_{\tau} = \big[ \bm{x}_{\tau}^{\top}\ \ \bm{y}_{\tau}^{\top} \big]^{\top} \in \mathbb{R}^{n}$, with $n = n_x + m = \sum_{k \in \mathcal{K}} m/k$.

By~\eqref{eq:taggcycle}, every aggregated entry is a fixed linear combination of $\bm{y}_{\tau}$, so the aggregated levels satisfy $\bm{x}_{\tau} = \bm{A}\bm{y}_{\tau}$. The matrix $\bm{A} \in \mathbb{R}^{n_x \times m}$ is built level by level: collecting the weights of level $k$ in the vector $\bm{w}^{[k]} = \big(w^{[k]}_1,\dots,w^{[k]}_k\big)^{\top} \in \mathbb{R}^{k}$, the non-overlapping aggregation of $\bm{y}_{\tau}$ at level $k$ is carried out by
\begin{equation}
\bm{A}^{[k]} = \bm{I}_{M_k} \otimes \big(\bm{w}^{[k]}\big)^{\top} \in \mathbb{R}^{M_k \times m},
\label{eq:Ak}
\end{equation}
whose $s$-th row places the weights $\big(\bm{w}^{[k]}\big)^{\top}$ on the $k$ high-frequency observations of the $s$-th block and zeros elsewhere. Stacking these blocks from the most to the least aggregated level gives
\begin{equation}
\bm{A} = \begin{bmatrix}
\bm{A}^{[k_p]}\\
\bm{A}^{[k_{p-1}]}\\
\vdots\\
\bm{A}^{[k_2]} \end{bmatrix}.
\label{eq:Amat}
\end{equation}
It's worth noting that temporal sums correspond to $\bm{w}^{[k]} = \bm{1}_k$, so that $\bm{A}^{[k]} = \bm{I}_{M_k} \otimes \bm{1}_k^{\top}$, and temporal averages to $\bm{w}^{[k]} = \tfrac{1}{k}\bm{1}_k$. At the most aggregated level $k_p = m$ one has $M_m = 1$, and $\bm{A}^{[m]} = \big(\bm{w}^{[m]}\big)^{\top}$ is a single row spanning the whole cycle. This leads to the two equivalent representations of the temporal hierarchy used throughout. The structural representation
\begin{equation}
\bm{z}_{\tau} = \bm{S}\bm{y}_{\tau},
\qquad \bm{S} = \begin{bmatrix} \bm{A} \\ \bm{I}_m \end{bmatrix} \in \mathbb{R}^{n \times m},
\label{eq:struct}
\end{equation}
expresses the whole hierarchy as a function of the high-frequency level, while the zero-constrained representation
\begin{equation}
\bm{C} \bm{z}_{\tau} = \bm{0}_{(n_x \times 1)},
\qquad \bm{C} = \big[ \bm{I}_{n_x}\ \ {-\bm{A}} \big] \in \mathbb{R}^{n_x \times n},
\label{eq:zero}
\end{equation}
states that each low-frequency level equals the corresponding combination of the high-frequency level \citep{Athanasopoulos2017-zh, Girolimetto2024-qe}.

When the number of high-frequency observations per cycle is constant, as in the $m = 24$ hourly-to-daily hierarchy considered here, the matrices $\bm{A}$, $\bm{S}$ and $\bm{C}$ are the same for every cycle. More generally the cycle length can vary, as when a daily series is aggregated to calendar months (where $m$ ranges from 28 to 31), and these matrices, together with the dimensions $M_k$, $n$ and $n_x$, become cycle-specific, $\bm{A}_\tau$, $\bm{S}_\tau$ and $\bm{C}_\tau$. Being fully determined by the calendar, they remain deterministic and known a priori once the cycle $\tau$ is fixed, so all subsequent developments carry over after replacing $\bm{A}$, $\bm{S}$ and $\bm{C}$ with their cycle-specific counterparts.

\subsection{Multiple base forecasts} \label{sec:base}

Forecasting a target cycle requires a forecast of its whole temporal hierarchy, that is, of all the entries of the coherent vector $\bm{z}_\tau$: the $m$ high-frequency values, the $M_k$ values of each aggregated level $k$, down to the single most aggregated value, for a total of $n$ entries. Rather than relying on a single source of base forecasts, suppose that $Q$ experts (e.g. competing forecasting models), indexed by $q = 1,\dots,Q$, are available, each producing a complete forecast of the hierarchy. To lighten the notation we suppress the cycle index $\tau$ and refer to a generic target vector $\bm{z} \in \mathbb{R}^{n}$. We focus on the balanced case, the one relevant to our application, in which every expert forecasts the entire hierarchy: expert $q$ delivers the base forecast $\widehat{\bm{z}}_q \in \mathbb{R}^{n}$ of $\bm{z}$. The extension to the unbalanced case, where each expert may forecast only some of the temporal levels, is straightforward following \citet{Girolimetto2024-qe}. Since the experts are estimated independently and the forecasts produced for different temporal levels need not satisfy the aggregation relations, the base forecasts are in general incoherent, $\bm{C}\widehat{\bm{z}}_q \neq \bm{0}_{(n_x \times 1)}$. We arrange the base forecasts in the $(n \times Q)$ matrix
\begin{equation}
\widehat{\bm{Z}} = \big[ \widehat{\bm{z}}_1\ \ \dots\ \ \widehat{\bm{z}}_Q \big] \in \mathbb{R}^{n \times Q},
\label{eq:Zmat}
\end{equation}
whose $q$-th column is the forecast of expert $q$, and whose $i$-th row gathers the $Q$ forecasts available for the $i$-th entry of $\bm{z}$. Stacking the experts one below the other gives the $(nQ \times 1)$ vector
\begin{equation}
\widehat{\bm{z}} = \mathrm{vec}\big(\widehat{\bm{Z}}\big) = \big[ \widehat{\bm{z}}_1^{\top}\ \ \dots\ \ \widehat{\bm{z}}_Q^{\top} \big]^{\top}.
\label{eq:zhat}
\end{equation}

We assume the base forecasts to be conditionally unbiased, $\mathrm{E}\big[\widehat{\bm{z}}_q - \bm{z} \mid \mathcal{I}\big] = \bm{0}$ for $q = 1,\dots,Q$, where $\mathcal{I}$ denotes the information set available at the forecast origin. To lighten the notation, in what follows all expectations are taken conditionally on $\mathcal{I}$ and we simply write $\mathrm{E}\big[\widehat{\bm{z}}_q\big] = \bm{z}$. Equivalently, the stacked forecasts satisfy $\mathrm{E}\big[\widehat{\bm{z}}\big] = \bm{K}\bm{z}$, where
\begin{equation}
\bm{K} = \bm{1}_Q \otimes \bm{I}_n \in \mathbb{R}^{nQ \times n}
\label{eq:K}
\end{equation}
replicates the target $Q$ times and $\bm{1}_Q$ is the $(Q \times 1)$ vector of ones. Collecting the base forecast errors $\widehat{\bm{\varepsilon}}_q = \widehat{\bm{z}}_q - \bm{z}$ into $\widehat{\bm{\varepsilon}} = \big[ \widehat{\bm{\varepsilon}}_1^{\top}\ \dots\ \widehat{\bm{\varepsilon}}_Q^{\top} \big]^{\top}$, we denote by
\begin{equation}
\bm{W} = \mathrm{E}\big[\widehat{\bm{\varepsilon}}\widehat{\bm{\varepsilon}}^{\top}\big] \in \mathbb{R}^{nQ \times nQ}
\label{eq:W}
\end{equation}
the positive-definite covariance matrix of the stacked base forecast errors. Its $(q,r)$ block $\bm{W}_{qr} = \mathrm{E}\big[\widehat{\bm{\varepsilon}}_q\widehat{\bm{\varepsilon}}_r^{\top}\big] \in \mathbb{R}^{n \times n}$ is the cross-covariance between the errors of experts $q$ and $r$: the diagonal blocks $\bm{W}_{qq}$ describe the error structure of each expert across temporal levels, while the off-diagonal blocks capture the dependence between experts. In this and the following section the matrix $\bm{W}$ is treated as known and its estimation, which plays a central role in the empirical analysis, is discussed in Section~\ref{sec:cov}.

\subsection{Coherent forecast combination} \label{sec:occ}

\begin{figure}[p]
\centering
\includegraphics[width = 0.8\linewidth]{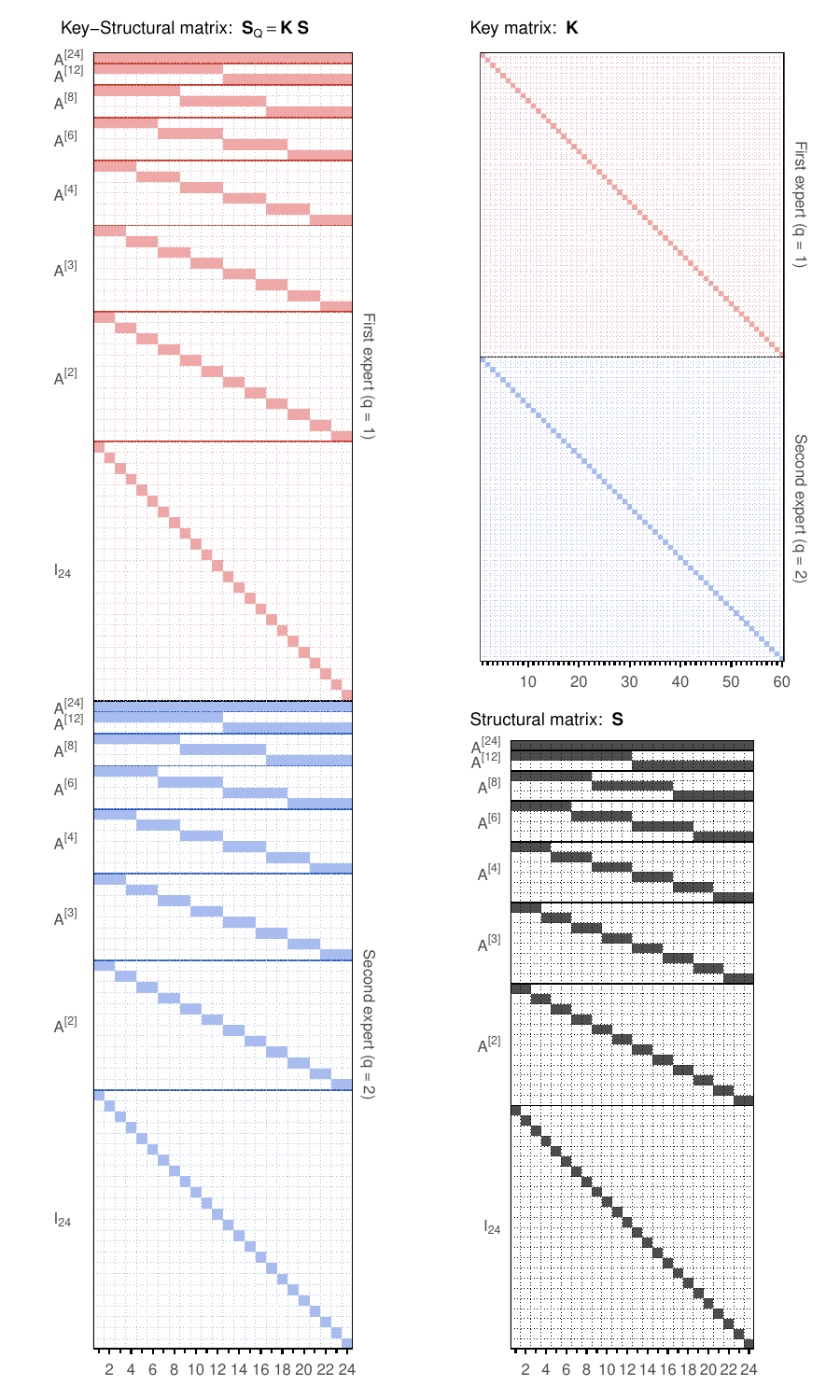}
\caption{The structural matrix $\bm{S}$, the key matrix $\bm{K} = \bm{1}_Q \otimes \bm{I}_n$ and their product $\bm{S}_Q = \bm{K}\bm{S}$ for the hourly hierarchy, with $m = 24$ and $n = 60$. For readability, $Q = 2$ experts are shown, distinguished by color.}
\label{fig:keymat}
\end{figure}

We now combine the $Q$ base forecasts into a single forecast that is coherent with the temporal hierarchy and, at the same time, exploits the information of all experts. Following the trace-minimization approach of \citet{Wickramasuriya2019-hs}, we look for a coherent forecast that is a linear function of the stacked base forecasts and has minimum error variance among all unbiased ones.

Any coherent forecast can be written in structural form as $\widetilde{\bm{z}} = \bm{S}\widetilde{\bm{y}}$, with $\widetilde{\bm{y}} \in \mathbb{R}^m$ the forecast vector of the highest-frequency time series. We take $\widetilde{\bm{y}} = \bm{G}\widehat{\bm{z}}$ linear in the base forecasts with $\bm{G} \in \mathbb{R}^{m \times nQ}$, then $\widetilde{\bm{z}} = \bm{S}\bm{G}\widehat{\bm{z}}$, which is coherent by construction, since $\bm{C}\bm{S} = \bm{0}$. Using $\mathrm{E}[\widehat{\bm{z}}] = \bm{K}\bm{z} = \bm{K}\bm{S}\bm{y}$, the forecast is unbiased for every $\bm{y}$ if and only if $\bm{G}\bm{S}_Q = \bm{I}_m$, where $\bm{S}_Q = \bm{K}\bm{S} = \bm{1}_Q \otimes \bm{S} \in \mathbb{R}^{nQ \times m}$ is the structural matrix replicated across the $Q$ experts, displayed in \autoref{fig:keymat} together with $\bm{K}$ and $\bm{S}$ for the hourly hierarchy. Under unbiasedness, the combined forecast error is $\widetilde{\bm{z}} - \bm{z} = \bm{S}\bm{G}\widehat{\bm{\varepsilon}}$, with covariance $\bm{S}\bm{G}\bm{W}\bm{G}^{\top}\bm{S}^{\top}$. We look for the matrix $\bm{G}$ that minimizes the trace of this covariance subject to the unbiasedness restriction $\bm{G}\bm{S}_Q = \bm{I}_m$  \citep{Wickramasuriya2019-hs}.

\begin{corollary}[Temporal coherent combination; specialization of Theorem~1 in \citealp{Girolimetto2024-qe}] \label{thm:occ}
Let $\bm{W}$ be the positive-definite covariance matrix of the base forecast errors. The optimal coherent forecast combination $\widetilde{\bm{z}}^{c} = \bm{S}\bm{G}\widehat{\bm{z}}$, obtained using the matrix $\bm{G}$ that minimizes $\mathrm{tr}\big(\bm{S}\bm{G}\bm{W}\bm{G}^{\top}\bm{S}^{\top}\big)$ subject to $\bm{G}\bm{S}_Q = \bm{I}_m$, is given by
\begin{equation}
\widetilde{\bm{z}}^{c} = \bm{S}\big(\bm{S}_Q^{\top}\bm{W}^{-1}\bm{S}_Q\big)^{-1}\bm{S}_Q^{\top}\bm{W}^{-1}\widehat{\bm{z}},
\label{eq:occ}
\end{equation}
or equivalently by
\begin{equation}
\widetilde{\bm{z}}^{c} = \Big[\bm{I}_n - \bm{W}_c\bm{C}^{\top}\big(\bm{C}\bm{W}_c\bm{C}^{\top}\big)^{-1}\bm{C}\Big]\bm{W}_c\bm{K}^{\top}\bm{W}^{-1}\widehat{\bm{z}},
\label{eq:occproj}
\end{equation}
with $\bm{W}_c = \big(\bm{K}^{\top}\bm{W}^{-1}\bm{K}\big)^{-1}$.
\end{corollary}
\begin{proof}
The result follows from \citet{Girolimetto2024-qe} by setting their constraint matrix equal to the temporal zero-constrained matrix $\bm{C}$ in~\eqref{eq:zero}. For completeness, a self-contained derivation in the present notation is given in \ref{app:occ}.
\end{proof}

The two equivalent expressions are the temporal counterparts of the structural and zero-constrained representations of \citet{Girolimetto2024-qe}, from which they differ only in the constraint matrix. Equation~\eqref{eq:occ} is obtained by minimizing the trace of the combined-forecast error covariance over all unbiased linear combinations of the base forecasts. Equation~\eqref{eq:occproj} solves instead a linearly constrained least squares adjustment \citep{Stone1942-pn, Byron1978-zb, Byron1979-pv}. Moreover, with a single expert ($Q = 1$) we have $\bm{S}_Q = \bm{S}$ and $\widehat{\bm{z}} = \widehat{\bm{z}}_1$, and~\eqref{eq:occ} reduces to the optimal temporal reconciliation of a single set of base forecasts $\widetilde{\bm{z}} = \bm{S}(\bm{S}^{\top}\bm{W}^{-1}\bm{S})^{-1}\bm{S}^{\top}\bm{W}^{-1}\widehat{\bm{z}}_1$ as proposed in \citet{Athanasopoulos2017-zh}. The optimal coherent combination is therefore the natural extension of optimal temporal reconciliation to multiple experts. The covariance matrix $\bm{W}$ drives both the combination weights and the reconciliation, and its estimation is examined in Section~\ref{sec:cov}.

\subsection{Estimation of the covariance matrix} \label{sec:cov}

The optimal combination of \autoref{thm:occ} takes the positive-definite matrix $\bm{W}$ as given, but in practice $\bm{W}$ is unknown and must be estimated. Since the errors of the target cycle are not observable, $\bm{W}$ is estimated from the in-sample errors, the differences between the forecasts produced by each model over its training sample and the corresponding observed values \citep{Hyndman2011-oc, Wickramasuriya2019-hs, Lipiecki2026-nm}. 

Mirroring the organization of the base forecasts in \eqref{eq:Zmat} and \eqref{eq:zhat}, let
\begin{equation}
\widehat{\bm{E}}_\tau = \big[ \widehat{\bm{\varepsilon}}_{1,\tau}\ \ \dots\ \ \widehat{\bm{\varepsilon}}_{Q,\tau} \big] =
\begin{bmatrix}
\widehat{\bm{\eta}}_{1,\tau}^{\top}\\
\vdots\\
\widehat{\bm{\eta}}_{n,\tau}^{\top}
\end{bmatrix}
\in \mathbb{R}^{n \times Q},
\qquad \tau = 1,\dots,N,
\label{eq:Emat}
\end{equation}
be the in-sample error matrix of cycle $\tau$: its $q$-th column $\widehat{\bm{\varepsilon}}_{q,\tau} = \widehat{\bm{z}}_{q,\tau} - \bm{z}_\tau \in \mathbb{R}^{n}$ contains the errors of expert $q$ over the whole temporal structure, while its $i$-th row $\widehat{\bm{\eta}}_{i,\tau}^{\top} \in \mathbb{R}^{1 \times Q}$ contains the errors of the $Q$ experts on the $i$-th element of $\bm{z}_\tau$, identified by a temporal level and a within-cycle position and called a node in what follows. Stacking the columns in $\widehat{\bm{\varepsilon}}_\tau = \mathrm{vec}\big(\widehat{\bm{E}}_\tau\big)$, the covariance matrix in~\eqref{eq:W} is estimated by its sample counterpart
\begin{equation}
\widehat{\bm{W}} = \frac{1}{N}\sum_{\tau=1}^{N}\widehat{\bm{\varepsilon}}_\tau\widehat{\bm{\varepsilon}}_\tau^{\top}.
\label{eq:samplecov}
\end{equation}
The sample covariance is a reliable estimate of $\bm{W}$ only when the dimension $nQ$ is small relative to $N$. The matrix $\widehat{\bm{W}}$ has $nQ(nQ+1)/2$ distinct entries, it is singular whenever $N < nQ$ and, even when invertible, its eigenvalues are systematically over-dispersed relative to those of $\bm{W}$, the largest being over-estimated and the smallest under-estimated \citep{ledoit2020analytical}. 
A reliable estimation strategy therefore combines two choices: a correlation structure, fixing which entries of $\bm{W}$ are estimated and which are set to zero, and a shrinkage technique, regularizing the estimation of the remaining entries. 

\begin{table}[!tb]
\centering
\caption{Covariance estimators obtained by combining a shrinkage technique (rows) with a correlation structure (columns). Each abbreviation is the technique with the structure as a subscript, no subscript denoting the full structure. In the block structures, the technique is applied separately to each diagonal block, $\widehat{\bm{W}}_{qq}$ in~\eqref{eq:be} for by-expert and $\widehat{\bm{\Gamma}}_{i}$ in~\eqref{eq:bn} for by-node representations. With linear shrinkage the error variances (darker diagonal) are the sample variances and only the covariances (lighter entries) are shrunk, whereas nonlinear shrinkage acts on the whole spectrum and also modifies the variances. For each structure, the header reports the sparsity index $I_S$ in~\eqref{eq:sparsity}, the share of entries of $\bm{W}$ set to zero, and the number of free parameters to be estimated: in the application of Section~\ref{sec:setup}, where $n = 60$ and $Q = 4$, $I_S$ equals $0$, $0.75$, $0.983$ and $0.996$, respectively.}
\label{tab:cov-abbrev}
\small
\resizebox{\linewidth}{!}{\begin{tabular}{@{}>{\raggedright\arraybackslash}p{0.51\linewidth}|cccc@{}}
\toprule
\raggedleft\arraybackslash
\parbox[b]{0.95\linewidth}{%
\raggedleft
Graphical representation for $Q = 4$ experts.\par
{\scriptsize
Shaded entries are estimated and the others are set to zero,
with the darker main diagonal marking the error variances,
which are always part of the estimate.\par
}\vskip0.5em}\makecell[br]{$\rightarrow$\\[1em] \phantom{.}}
& \resizebox{0.1\linewidth}{!}{\begin{tikzpicture}[x=0.16cm, y=-0.16cm]
  \fill[black!4] (0,0) rectangle (28,28);
  \fill[thiefblue!30] (0,0) rectangle (28,28);
  \foreach \s in {7,14,21}{\draw[dashed, black!50, line width=0.45pt] (\s,0) -- (\s,28) (0,\s) -- (28,\s);}
  \foreach \i in {1,...,28}{\fill[thiefblue!70] (\i-1,\i-1) rectangle (\i,\i);}
  \draw[black, line width=1pt] (0,0) rectangle (28,28);
\end{tikzpicture}} & \resizebox{0.1\linewidth}{!}{\begin{tikzpicture}[x=0.16cm, y=-0.16cm]
\fill[black!4] (0,0) rectangle (28,28);
\foreach \q in {0,...,3}{\fill[thiefblue!30] (\q*7, \q*7) rectangle (\q*7+7, \q*7+7);}
\foreach \s in {7,14,21}{\draw[dashed, black!50, line width=0.45pt] (\s,0) -- (\s,28) (0,\s) -- (28,\s);}
\foreach \i in {1,...,28}{\fill[thiefblue!70] (\i-1,\i-1) rectangle (\i,\i);}
\draw[black, line width=1pt] (0,0) rectangle (28,28);
\end{tikzpicture}} & \resizebox{0.1\linewidth}{!}{\begin{tikzpicture}[x=0.16cm, y=-0.16cm]
\fill[black!4] (0,0) rectangle (28,28);
\foreach \a in {0,...,3}{\foreach \b in {0,...,3}{\foreach \d in {1,...,7}{
\fill[thiefblue!30] (\a*7+\d-1, \b*7+\d-1) rectangle (\a*7+\d, \b*7+\d);}}}
\foreach \s in {7,14,21}{\draw[dashed, black!50, line width=0.45pt] (\s,0) -- (\s,28) (0,\s) -- (28,\s);}
\foreach \i in {1,...,28}{\fill[thiefblue!70] (\i-1,\i-1) rectangle (\i,\i);}
\draw[black, line width=1pt] (0,0) rectangle (28,28);
\end{tikzpicture}}  & \resizebox{0.1\linewidth}{!}{\begin{tikzpicture}[x=0.16cm, y=-0.16cm]
\fill[black!4] (0,0) rectangle (28,28);
\foreach \i in {1,...,28}{\fill[thiefblue!70] (\i-1,\i-1) rectangle (\i,\i);}
\foreach \s in {7,14,21}{\draw[dashed, black!50, line width=0.45pt] (\s,0) -- (\s,28) (0,\s) -- (28,\s);}
\draw[black, line width=1pt] (0,0) rectangle (28,28);
\end{tikzpicture}} \\
\raggedleft\arraybackslash sparsity index $\rightarrow$ &\scriptsize $0$ & \scriptsize $1-1/Q$ & \scriptsize $1-1/n$ & \scriptsize $1-1/(nQ)$\\
\raggedleft\arraybackslash free parameters $\rightarrow$& \scriptsize $\dfrac{nQ(nQ+1)}{2}$ & \scriptsize $\dfrac{Qn(n+1)}{2}$ & \scriptsize $\dfrac{nQ(Q+1)}{2}$ & \scriptsize $nQ$\\
\cmidrule(l){2-5}
\textbf{Shrinkage approach and references} & \textbf{Full} &  \textbf{By-expert} &  \textbf{By-node} &  \textbf{Diag.}\\
\midrule
None, variances only \linebreak \citet{Hyndman2016, Athanasopoulos2017-zh} & --- & --- & --- & $wls$\\
\addlinespace[3pt]
Linear~\eqref{eq:shr} \linebreak \citet{ledoit2004, Schafer2005-jh} & $shr$ & $shr_{be}$ & $shr_{bn}$ & ---\\
\addlinespace[3pt]
Non-linear~\eqref{eq:nlshr} with Stein's loss \linebreak \citet{ledoit2022power} & $lis$ & $lis_{be}$ & $lis_{bn}$ & ---\\
\addlinespace[3pt]
Non-linear~\eqref{eq:nlshr} with Frobenius loss of $\bm{W}^{-1}$ \linebreak \citet{ledoit2022quadratic} & $qis$ & $qis_{be}$ & $qis_{bn}$ & ---\\
\addlinespace[3pt]
Non-linear~\eqref{eq:nlshr} with symmetrized Kullback--Leibler loss \linebreak \citet{ledoit2022power} & $gis$ & $gis_{be}$ & $gis_{bn}$ & ---\\
\bottomrule
\end{tabular}}
\end{table}

We first consider the four correlation structures shown in the first row of \autoref{tab:cov-abbrev}, each corresponding to a different pattern of zero constraints on the covariance matrix. At one extreme, the diagonal structure sets all covariances to zero and retains only the $nQ$ error variances,
\begin{equation}
\widehat{\bm{W}}_{\mathrm{wls}} = \widehat{\bm{W}} \odot \bm{I}_{nQ},
\label{eq:wls}
\end{equation}
where $\odot$ denotes the Hadamard product. The resulting estimator is positive definite by construction and requires no regularization. At the opposite side, the full structure leaves $\widehat{\bm{W}}$ unrestricted and is the only specification that accounts for correlations between different experts and different nodes.

Between these extremes lie two block-diagonal structures. The by-expert ($\text{be}$) structure retains the diagonal blocks of $\widehat{\bm{W}}$,
\begin{equation}
\widehat{\bm{W}}_{\text{be}} = \operatorname{diag}\big(\widehat{\bm{W}}_{11}, \dots, \widehat{\bm{W}}_{QQ}\big),
\qquad
\widehat{\bm{W}}_{qq} = \frac{1}{N}\sum_{\tau=1}^{N}\widehat{\bm{\varepsilon}}_{q,\tau}\widehat{\bm{\varepsilon}}_{q,\tau}^{\top},
\label{eq:be}
\end{equation}
preserving each expert’s error dependence across temporal levels while assuming independence between experts. The by-node (bn) structure instead groups the errors of all experts at the same node. Let $\bm{P}$ denote the $(nQ\times nQ)$ commutation matrix \citep{Magnus2019-mn}, such that $\bm{P}\widehat{\bm{\varepsilon}}_\tau=\mathrm{vec}(\widehat{\bm{E}}_\tau^\top)$. Then,
\begin{equation}
\widehat{\bm{W}}_{\text{bn}} =
\bm{P}^{\top}
\operatorname{diag}\big(\widehat{\bm{\Gamma}}_{1}, \dots, \widehat{\bm{\Gamma}}_{n}\big)
\bm{P},
\qquad
\widehat{\bm{\Gamma}}_{i} =
\frac{1}{N}\sum_{\tau=1}^{N}\widehat{\bm{\eta}}_{i,\tau}\widehat{\bm{\eta}}_{i,\tau}^{\top},
\label{eq:bn}
\end{equation}
which preserves the dependence among experts at each node while assuming independence across nodes. The two block structures extend to the temporal setting the by-expert and by-variable organizations of \citet{Girolimetto2024-qe}, but we use the term by-node because temporal hierarchies involve a single variable observed at different aggregation levels.

The parsimony of the four structures can be compared through their matrix sparsity, defined as the proportion of entries constrained to zero:
\begin{equation}
I_S = 1 - \frac{\nu}{(nQ)^{2}},
\label{eq:sparsity}
\end{equation}
where $\nu$ is the number of nonzero entries. The full structure keeps all $\nu = (nQ)^{2}$ entries, so $I_S = 0$. The by-expert structure keeps $\nu = Qn^{2}$ entries and the by-node structure $\nu = nQ^{2}$, giving $I_S = 1 - 1/Q$ and $I_S = 1 - 1/n$, while the diagonal structure keeps only $\nu = nQ$ entries, so $I_S = 1 - 1/(nQ)$. \autoref{tab:cov-abbrev} reports, together with $I_S$, the number of free parameters of each structure. In the application, where $n = 60$ and $Q = 4$, the by-expert structure sets $75\%$ of the entries to zero and leaves $7\,320$ parameters to be estimated, against the $28\,920$ of the full structure, while the by-node and the diagonal structures set $98.3\%$ and $99.6\%$ of the entries to zero, with $600$ and $240$ parameters, respectively.

Even after imposing a correlation structure, the remaining nonzero entries must still be estimated from the same $N$ cycles and, for the full and by-expert structures, their dimension remains large relative to $N$. We therefore regularize the sample covariance by shrinkage, which trades a small bias for a substantial reduction in estimation variance and yields a well-conditioned, positive-definite estimator. We consider one linear and three nonlinear shrinkage techniques:

\begin{itemize}
\item \emph{Linear shrinkage} \citep{ledoit2004} use a convex combination of covariance and a structured target, here the diagonal of $\widehat{\bm{W}}$,
\begin{equation}
\widehat{\bm{W}}_{\text{shr}} = 
\widehat{\lambda}\big(\widehat{\bm{W}}\odot\bm{I}_{nQ}\big)
+
(1-\widehat{\lambda})\widehat{\bm{W}},
\label{eq:shr}
\end{equation}
where the shrinkage intensity $\widehat{\lambda}$ is estimated analytically following \citet{Schafer2005-jh} and truncated to the unit interval. For $\widehat{\lambda} \in (0,1]$, \eqref{eq:shr} is a convex combination of a positive-definite diagonal matrix and a positive semi-definite matrix, and is therefore positive definite even when $N < nQ$. It is widely used in forecast reconciliation starting from \cite{Wickramasuriya2019-hs};

\item \emph{Nonlinear shrinkage} \citep{ledoit2022power} regularizes the spectrum of the sample covariance matrix while preserving its eigenvectors. Given the spectral decomposition $\widehat{\bm{W}}= \displaystyle\sum_{i=1}^{nQ}\widehat{\ell}_i\bm{u}_i\bm{u}_i^\top$, the estimator takes the form
\begin{equation}
\widehat{\bm{W}}_{\text{nl}} = \sum_{i=1}^{nQ} \widehat{\delta}_i \bm{u}_i \bm{u}_i^\top,
\label{eq:nlshr}
\end{equation}
replacing each sample eigenvalue $\widehat{\ell}_i$ with an individually shrunk counterpart $\widehat{\delta}_i$ computed analytically through \citet{ledoit2020analytical}. 
We consider the three estimators reviewed by \citet{ledoit2022power}: linear-inverse shrinkage (lis), quadratic-inverse shrinkage \citep[qis, see also][]{ledoit2022quadratic}, and geometric-inverse shrinkage (gis).
\end{itemize}

Each correlation structure can be combined with any shrinkage technique. Shrinkage is applied either to the full covariance matrix or, blockwise, to the diagonal blocks $\widehat{\bm{W}}_{qq}$ in \eqref{eq:be} and $\widehat{\bm{\Gamma}}_{i}$ in \eqref{eq:bn}, with the shrinkage parameters estimated separately for each block. Since a block-diagonal matrix is positive definite if and only if its blocks are, every resulting estimator is a valid input for \autoref{thm:occ}. The resulting thirteen estimators are summarized in \autoref{tab:cov-abbrev}, where the abbreviations combine the shrinkage technique with the correlation structure as a subscript, while no subscript denotes the full structure. The by-expert estimator with linear shrinkage, $shr_{be}$, is used as the benchmark in Section~\ref{sec:main}, whereas all thirteen estimators are compared in Section~\ref{sec:robust}.

\section{Forecasting day-ahead electricity prices}\label{sec:exp}

\subsection{Experimental setup} \label{sec:setup}

We evaluate the proposed methodology on the day-ahead electricity price dataset of \citet{Lipiecki2026-nm}, which is naturally observed at multiple temporal frequencies and on which the value of temporal reconciliation has already been established. We consider the same two major European markets, the German EPEX-DE and the Spanish OMIE, and the four-year out-of-sample period spanning 2021 to 2024. Crucially, we do not re-estimate any forecasting model: the base forecasts at all temporal levels are taken directly from the replication package of \citet{Lipiecki2026-nm}\footnote{Data and base forecasts are publicly available at \url{https://github.com/lipiecki/thief}.}, and enter our framework unchanged as the inputs to be combined. 

The day-ahead market is settled on a daily cycle of $m = 24$ hourly prices, the highest available frequency. The set of factors $\mathcal{K} = \{1,2,3,4,6,8,12,24\}$ defines eight temporal levels: the 24 hourly prices, the non-overlapping blocks of $2$, $3$, $4$, $6$, $8$ and $12$ hours, and the baseload, that is the average price of the day. Since block and baseload products are arithmetic means of the underlying hourly prices, the aggregation weights are $w^{[k]}_i = 1/k$ (see Section~\ref{sec:th}). 

The base forecasts come from the four model classes of \citet{Lipiecki2026-nm}, deliberately chosen to span a wide range of complexity, statistical assumptions and economic content. The first, ARX, is a parsimonious autoregressive model with exogenous regressors estimated by linear regression. The second, NARX, is its nonlinear counterpart, combining the forecasts of several feedforward neural networks to reduce estimation noise. The third, XGB, is a fully nonparametric learner based on gradient-boosted decision trees, with hyperparameters tuned on the calibration window. The fourth, MITRA, is a large pretrained tabular transformer foundation model, applied day by day in a rolling-window fashion and representing the recent wave of foundation models in forecasting. The careful design and estimation of these forecasts is a central contribution of \citet{Lipiecki2026-nm}, who devote most of their analysis to it, and we refer the reader to their work for the full specification of each model and of its input features. What matters for the present study is that the four models embody genuinely different views of the same variable of interest (prices), so that combining them is informative, and that, for each market, every model produces base forecasts independently at all temporal levels, which are in general incoherent. This places us in the balanced case of Section~\ref{sec:base}, with $Q = 4$ experts each forecasting the whole hierarchy ($n = 60$). 

The base forecasts are produced day by day under the rolling scheme: each model is calibrated on a three-year window of past observations, rolled forward by one day at each step, with the initial window ending on 31 December 2020 and the evaluation running over the subsequent four years. We focus on the post-forecasting step of enforcing the temporal constraints, in two stages. First, we replicate the single-expert temporal reconciliation of \citet{Lipiecki2026-nm} for each of the four models, re-implementing it in \textsf{R} with the \pkgreco{} package, whereas their analysis is carried out in \textsf{Julia}. This replication validates our pipeline and provides the single-expert benchmark. Second, we extend the procedure to the whole set of experts through the coherent combination of Section~\ref{sec:occ}, implemented in the \pkgcomb{} package. In both stages the covariance matrix $\bm{W}$ is estimated from the in-sample base forecast errors over the calibration window, using the estimators of Section~\ref{sec:cov}.

Forecast accuracy is assessed by the mean absolute error (MAE) and the root mean squared error (RMSE), computed separately at each temporal level. For level $k$, let $\check{x}^{[k]}_{\tau,s}$ be the forecast of the approach under evaluation, base, reconciled or combined, for the $s$-th observation of day $\tau$, let $x^{[k]}_{\tau,s}$ be the realized value, $D$ the number of test days, and $M_k = 24/k$ the number of level-$k$ observations per day. Then
\begin{equation}
\mathrm{MAE}^{[k]} = \frac{1}{DM_k}\sum_{\tau=1}^{D}\sum_{s=1}^{M_k}\big| x^{[k]}_{\tau,s} - \check{x}^{[k]}_{\tau,s}\big|,
\qquad
\mathrm{RMSE}^{[k]} = \sqrt{\frac{1}{DM_k}\sum_{\tau=1}^{D}\sum_{s=1}^{M_k}\big( x^{[k]}_{\tau,s} - \check{x}^{[k]}_{\tau,s}\big)^{2}}.
\label{eq:metrics}
\end{equation}
In addition, statistical significance is assessed using the Diebold–Mariano (DM) test \citep{Diebold1995-xs} for pairwise forecast comparisons and the Multiple Comparison with the Best (MCB) Nemenyi procedure \citep{Koning2005-hv} as implemented in the \texttt{tsutils} package \citep{tsutils}. Following the multivariate version of the DM test \citep{Ziel2018-dm, Lago2021-dp}, for each temporal level $k$ the loss differential is computed on the daily average of the absolute or squared errors over the $M_k$ observations of the day, which yields one observation per test day and avoids treating the strongly dependent intraday errors as independent. The test is one-sided, with the alternative hypothesis that the second forecast in the comparison is more accurate than the first, and its variance is estimated by the sample variance of the daily loss differential, with the small-sample correction of \citet{Harvey1997-dm} as implemented in the \texttt{forecast} package \citep{forecast}.

\subsection{Main results} \label{sec:main}

\autoref{tab:main} reports\footnote{Results for all covariance estimators are presented in the robustness analysis of Section~\ref{sec:robust}} the MAE and RMSE of the base forecasts of the four experts, of their single-expert temporal reconciliation (subscript $r$, using the $shr$ estimator), and of the proposed coherent combination (subscript $c$, using the reference estimator $shr_{be}$), together with the percentage improvement of the latter two over the corresponding base forecasts for each temporal level and both markets. Since the combination pools the four experts, a single combined forecast is available at each level, and its percentage improvement is computed with respect to each expert in turn. The significance of the accuracy differences is assessed by the DM test \citep{Diebold1995-xs} at the $1\%$ level.

\begin{table}[!tb]
\centering
\linespread{1}
\small
\setlength{\tabcolsep}{2pt}
\caption{MAE and RMSE~\eqref{eq:metrics} at each temporal level for the German (EPEX-DE) and Spanish (OMIE) day-ahead markets over the 2021--2024 test period. For each expert, the columns report the accuracy of the base forecasts, of the reconciled forecasts (subscript $r$, $shr$ estimator) and of the coherent combination of the four experts (subscript $c$, $shr_{be}$ estimator), together with the percentage improvements of the latter two over the base forecasts.}
\label{tab:main}
\resizebox{\linewidth}{!}{
\input{./tables/shrbe_tab_full_dm.tex}}
\raggedright
{\footnotesize \vskip0.25em\textbf{Note:} Values in bold indicate the best percentage improvement between reconciliation and the coherent combination. Superscripts indicate the results of the one-sided multivariate DM test on daily average losses at the $1\%$ level: $^1$ not significantly more accurate than the base forecasts, $^2$ not significantly more accurate than the reconciled forecasts.}
\end{table}

The single-expert results replicate, with an independent implementation, the central finding of \citet{Lipiecki2026-nm}: temporal reconciliation improves upon the base forecasts for every model, at every temporal level and in both markets. MITRA provides the most accurate base and reconciled forecasts at every level, ARX the least accurate, and the cases in which reconciliation does not significantly improve upon the base forecasts (superscript~$1$) almost exclusively involve the Spanish market, mainly for NARX and MITRA, whose base forecasts are already close in accuracy to their reconciled counterparts.

By exploiting information from multiple experts, the temporal coherent combination consistently outperforms single-expert reconciliation. At the hourly level, where the day-ahead market is settled, it achieves the highest accuracy in both markets. The magnitude of the improvement depends on the quality of the individual expert: the largest gains are observed for the weakest models, although even the strongest expert benefits from replacing its reconciled forecasts with the combined forecast in almost every case. As with reconciliation, the improvements generally become larger at higher temporal aggregation levels, making coherent combination particularly valuable for the frequencies relevant to block and baseload trading. 

The Diebold–Mariano test qualifies these findings. The combined forecasts are significantly more accurate than the base forecasts of every expert, at every temporal level and in both markets, with the single exception of the RMSE of MITRA at the 12H level in Germany. They are also significantly more accurate than the reconciled forecasts of ARX, NARX and XGB in all cases. The comparison with reconciled MITRA, the strongest single-expert benchmark, is more nuanced. In the Spanish market the combination significantly improves upon it at every level in terms of MAE, and up to the 6H level in terms of RMSE. In the German market the improvement is significant only for the hourly MAE, while at the other levels the two approaches are statistically indistinguishable, with reconciled MITRA attaining a marginally lower error at the 12H and 24H levels. The coherent combination is therefore never significantly worse than the best reconciled expert, and it delivers this accuracy without requiring the forecaster to identify the best expert in advance.

\begin{figure}[!tb]
\centering
\includegraphics[width = \linewidth]{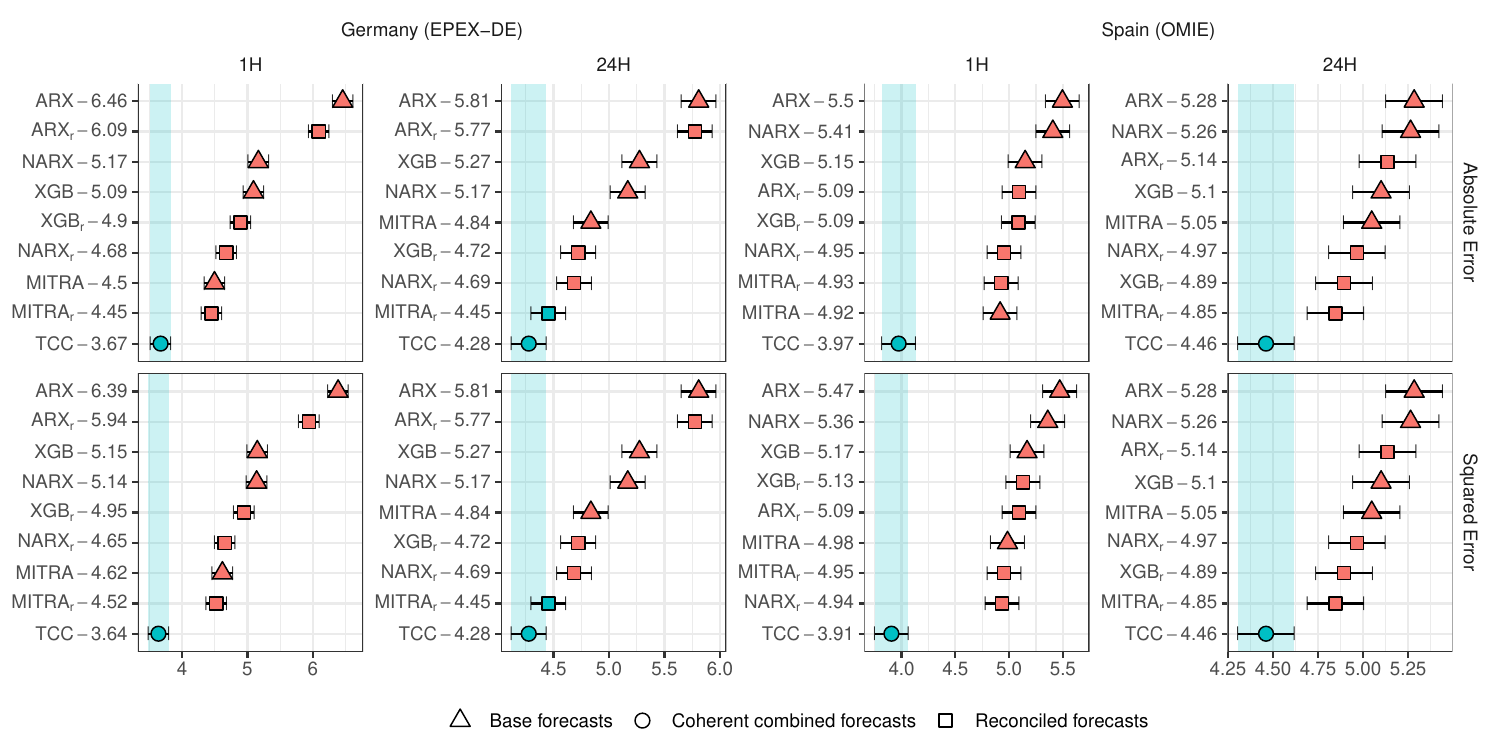}
\caption{Average ranks and MCB Nemenyi test intervals, computed on the absolute (top) and squared (bottom) errors over the $1\,461$ test days, at the hourly (1H) and daily (24H) levels for the German and Spanish markets. TCC denotes the temporal coherent combination ($shr_{be}$) and the subscript $r$ the reconciled forecasts ($shr$). The shaded area marks the confidence region of the best-ranked approach: approaches whose intervals do not overlap it perform significantly worse.}
\label{fig:mcb}
\end{figure}

\autoref{fig:mcb} summarizes this evidence through the MCB Nemenyi test \citep{Koning2005-hv}: for the hourly and daily levels of both markets, each panel displays the ranks of the nine approaches (four base, four reconciled, and the coherent combination, denoted TCC), averaged across the test days on absolute and squared errors. TCC achieves the lowest average rank in every panel and is the sole approach in the confidence region of the best in all cases but the daily level of the German market, where reconciled MITRA is statistically closed to the temporal combination, in line with the DM results of \autoref{tab:main}. The comparison extended to all eight temporal levels, reported in \ref{app:mcb}, confirms this picture, with the confidence region of the best including reconciled MITRA only at the 12- and 24-hour levels of the German market.

\section{Robustness analysis} \label{sec:robust}

The empirical results presented in Section~\ref{sec:main} are based on three modelling choices: the use of the $shr_{be}$ covariance estimator, a pool of four experts, and the simultaneous estimation of forecast combination and reconciliation. This section evaluates the robustness of the results with respect to each of these choices. Sections~\ref{sec:covariances} and~\ref{sec:shr} examine the role of the correlation structure and the shrinkage technique, respectively, Section~\ref{sec:loo} assesses the contribution of each expert through a leave-one-out analysis, and Section~\ref{sec:seq} compares the proposed simultaneous approach with the sequential strategy that first combines forecasts and then reconciles them. Unless otherwise stated, results are reported in terms of MAE, however the corresponding RMSE results, reported in \ref{app:rmse}, lead to the same conclusions.

\subsection{Alternative covariance estimators} \label{sec:covariances}

\begin{table}[!tb]
\centering
\linespread{1}
\small
\setlength{\tabcolsep}{3pt}
\caption{MAE at each temporal level for the base forecasts, the single-expert reconciled forecasts with the $wls$ and $shr$ estimators, and the temporal coherent combination with the $wls$, $shr$, $shr_{bn}$ and $shr_{be}$ estimators of \autoref{tab:cov-abbrev}. For each market and temporal level, bold and italics denote the best and the second-best approach, respectively, ranked on unrounded values.}
\label{tab:covariance}
\resizebox{\linewidth}{!}{
\input{./tables/tab_covariance_mae.tex}}
\end{table}

\autoref{tab:covariance} investigates the role of the correlation structure. The base forecasts and the $shr$ rows of the upper panels, as well as the $shr_{be}$ row of the bottom panel, coincide with the MAE values in \autoref{tab:main} and are repeated here to ease the comparison. The upper panels compare the single-expert reconciled forecasts obtained with the diagonal ($wls$) and linear shrinkage ($shr$) estimators, while the bottom panel reports the coherent combination under the four covariance structures of \autoref{tab:cov-abbrev}, using linear shrinkage whenever required.

For single-expert reconciliation, the $shr$ estimator consistently improves upon $wls$ across all experts, temporal levels, and markets, confirming that the correlation of forecast errors across temporal levels contains useful information. Reconciled MITRA with $shr$ consequently provides the strongest single-expert benchmark. The comparison is even more informative for coherent combination. The by-expert structure ($shr_{be}$) is the most accurate specification at every temporal level in both markets, the only exceptions being the 12H and 24H levels in Germany discussed in Section~\ref{sec:main}. The diagonal structure ranks second in Spain, where it already outperforms every reconciled expert, but in Germany it never matches reconciled MITRA. The remaining two structures ($shr$ and $shr_{bn}$) perform substantially worse. Overall, the results indicate that preserving the temporal dependence within each expert while treating different experts as conditionally independent provides the most effective balance between statistical efficiency and estimation error.

\subsection{The role of shrinkage} \label{sec:shr}

\autoref{tab:shrinkage} completes the comparison of the covariance estimators by fixing the correlation structure and varying the shrinkage technique.

\begin{table}[tb]
\centering
\linespread{1}
\small
\setlength{\tabcolsep}{3pt}
\caption{MAE of the temporal coherent combination for the linear ($shr$) and nonlinear ($lis$, $qis$, $gis$) shrinkage techniques of \autoref{tab:cov-abbrev}, each applied to the full matrix (no subscript), to the by-node blocks (subscript $bn$) and to the by-expert blocks (subscript $be$); the diagonal estimator $wls$ is reported in \autoref{tab:covariance}. For each market and temporal level, bold and italics denote the best and the second-best estimator, respectively, ranked on unrounded values; differences among the nonlinear estimators often appear only beyond the second decimal place.}
\label{tab:shrinkage}
\resizebox{\linewidth}{!}{\input{./tables/tab_shr_mae.tex}}
\end{table}

The choice of shrinkage has a much smaller impact than the choice of correlation structure. Under the by-expert specification, differences among the four estimators are negligible at every temporal level and in both markets, with the three nonlinear methods almost always agreeing to the second decimal place. Moreover, nonlinear shrinkage never outperforms its linear counterpart. This behaviour is explained by the concentration ratio $c = d/N$, where $d$ is the dimension of the matrix to be estimated and $N = 1\,085$ is the number of in-sample cycles available in each rolling window. For the by-expert blocks, $d = n = 60$ and $c \approx 0.055$. Random matrix theory shows that the dispersion of the sample eigenvalues around the population ones is governed by $c$: when the population covariance is the identity, the sample eigenvalues lie asymptotically in the interval $\big[(1-\sqrt{c})^2, (1+\sqrt{c})^2\big]$ \citep{ledoit2020analytical}, which for $c \approx 0.055$ is only $[0.58, 1.53]$. With such a mild over-dispersion, the optimal nonlinear correction of each eigenvalue is close to the uniform correction applied by linear shrinkage, and the three nonlinear estimators, which differ only in the loss function used to derive the shrunk eigenvalues and hence mainly in the treatment of the smallest ones, become nearly indistinguishable. The additional flexibility of nonlinear shrinkage is therefore expected to pay off only for larger concentration ratios. Even under the full covariance structure, where $d = nQ = 240$ and $c \approx 0.22$, linear shrinkage remains marginally more accurate, consistent with the evidence of Section~\ref{sec:covariances} that the poor performance of this structure is mainly due to the cross-expert correlations rather than to the estimation of the spectrum. Overall, the evidence indicates that the correlation structure is the primary determinant of forecasting performance, whereas the particular shrinkage technique plays only a secondary role, supporting the use of the simple and analytically available $shr_{be}$ estimator as the reference specification.

\subsection{Leaving out an expert} \label{sec:loo}

\begin{table}[tb]
\centering
\linespread{1}
\small
\setlength{\tabcolsep}{3pt}
\caption{MAE and RMSE at each temporal level of the temporal coherent combination ($shr_{be}$) when the expert in the first column is excluded from the pool ($Q = 3$); None denotes the full pool of four experts ($Q = 4$). For each market and temporal level, bold and italics denote the best and the second-best pool, respectively.}
\label{tab:loo}
\resizebox{\linewidth}{!}{\input{./tables/tab_split_shrbe.tex}}
\end{table}

The robustness of the coherent combination with respect to the composition of the expert pool is assessed through a leave-one-out analysis. \autoref{tab:loo} reports the MAE and RMSE obtained with the $shr_{be}$ estimator after excluding one expert at a time, together with the full-pool results of Section~\ref{sec:main}.

As expected, removing MITRA, the strongest expert, has the largest impact on forecast accuracy in both markets. Nevertheless, the resulting three-expert combination remains essentially as accurate as reconciled MITRA in Germany and still outperforms it in Spain, while clearly improving upon the reconciled forecasts of the remaining experts. At the opposite extreme, excluding ARX has only a negligible effect, indicating that most of its predictive information is already captured by the other models. On the other hand, excluding XGB leads to a slight but systematic improvement in Germany and a marginal reduction in RMSE in Spain. This suggests that, once the covariance matrix has to be estimated, an expert whose information largely overlaps with that of the remaining models may contribute more estimation noise than additional signal, so that enlarging the pool is not necessarily beneficial \citep{Timmermann2006}. Overall, however, the leave-one-out combinations remain remarkably stable. At the hourly level, every reduced pool performs at least as well as the best reconciled expert, the only exception being the pool without MITRA in the German market, which is statistically indistinguishable from reconciled MITRA.

\subsection{A first look at sequential alternatives: first combination, then reconciliation} \label{sec:seq}

The proposed method performs forecast combination and reconciliation simultaneously. Following \citet{Rostami-Tabar2024-qf} and \cite{Girolimetto2024-qe}, we consider a simple and natural sequential alternative that first combines the forecasts of the $Q$ experts at each node using the simple average and then reconciles the resulting forecasts using either the diagonal variance estimator ($wls$) or the linear shrinkage covariance estimator ($shr$). \autoref{tab:seq} compares the unreconciled simple average ($sa$), its two reconciled versions ($sa+wls$ and $sa+shr$), the best single-expert reconciliation (MITRA with $shr$), and the best-performing coherent combination ($shr_{be}$).

Reconciling the simple average ($sa+wls$ and $sa+shr$) consistently improves upon the unreconciled forecasts, making the sequential strategy a competitive benchmark. In the Spanish market it ranks second only to the coherent combination at every temporal level. Nevertheless, the proposed simultaneous approach remains the most accurate specification in both markets and at every level. The advantage is particularly pronounced in Germany, where the sequential procedure does not even outperform reconciled MITRA, whereas in Spain the differences are smaller. This result reflects the additional information exploited by the one-step estimator. Averaging before reconciliation discards both the relative accuracy of the individual experts and the dependence structure of their forecast errors, whereas the coherent combination incorporates both sources of information directly through the covariance matrix.

\begin{table}[tb]
\centering
\linespread{1}
\small
\setlength{\tabcolsep}{3pt}
\caption{MAE at each temporal level of the sequential first-combined-then-reconciled approaches, in which the equal-weighted average of the four experts ($sa$) is reconciled as a single expert with the $wls$ and $shr$ estimators, compared with the best single-expert reconciliation (MITRA base forecasts, $shr$ estimator) and the best temporal coherent combination ($shr_{be}$). For each market and temporal level, bold and italics denote the best and the second-best approach, respectively.}
\label{tab:seq}
\resizebox{\linewidth}{!}{\input{./tables/tab_seq_mae.tex}}
\raggedright
{\scriptsize \vskip0.25em\textbf{$^\ast$} Simple average combination is not coherent.}
\end{table}

\section{Conclusion} \label{sec:conclusion}

This paper has extended the coherent forecast combination framework of \citet{Girolimetto2024-qe} to temporal hierarchies. The optimal temporal coherent approach combines the forecasts that $Q$ experts produce into a single forecast that satisfies the temporal aggregation constraints and has minimum error variance among the linear unbiased combinations. The solution is available in closed form, in two equivalent structural and zero-constrained representations, and contains the optimal single-expert temporal reconciliation of \citet{Athanasopoulos2017-zh} as the special case $Q = 1$. Since the solution depends on the covariance matrix of the stacked base forecast errors, whose dimension grows with both the hierarchy and the number of experts, we have analysed different estimation strategy that crosses four correlation structures with linear and nonlinear shrinkage techniques.

The empirical analysis, based on the day-ahead electricity price dataset of \citet{Lipiecki2026-nm}, shows that temporal coherent combination with the by-expert covariance estimator significantly improves upon the base forecasts and upon all but the strongest reconciled expert at every temporal level in both the German and Spanish markets. With respect to the strongest reconciled expert, it is significantly more accurate in the Spanish market and at the hourly level in the German market, and statistically indistinguishable elsewhere, while never being significantly worse. The robustness analysis further shows that the key modelling decision is the choice of correlation structure: preserving the dependence of each expert’s forecast errors across temporal levels while treating different experts as uncorrelated consistently delivers the best performance, whereas the specific shrinkage technique plays only a secondary role.

Some directions for future research remain open. The present analysis is confined to point forecasts, and a probabilistic extension in the spirit of \citet{Girolimetto2024-jm} and \citet{Wickramasuriya2024-yv} would allow the combination to deliver coherent predictive distributions, which are increasingly required for risk management in electricity markets. A second direction concerns estimating the covariance matrix from out-of-sample validation errors, as proposed by \citet{Abolghasemi2025-fp}, rather than from in-sample errors as in standard reconciliation practice. The sequential strategy of Section~\ref{sec:seq} also deserves a deeper exploration. In our experiments the combination step relied on the equal-weighted average, and replacing it with estimated combination weights, or characterizing the conditions under which the two-step procedure attains an accuracy comparable to the simultaneous solution, would clarify how much of the documented gap can be recovered at a lower modeling cost. Finally, the approach could be embedded in a cross-temporal setting \citep{Di_Fonzo2023-dg, Girolimetto2024-jm}, combining experts across market zones and temporal granularities at once.

\appendix

\section{Proof of \autoref{thm:occ}} \label{app:occ}

\noindent The two expressions of \autoref{thm:occ} are the temporal counterparts of the zero-constrained form (Theorem~1) and of the equivalent structural form (Table~3) of \citet{Girolimetto2024-qe}. We prove each in turn and then establish their equivalence, adapting their argument to the present notation. Throughout, $\bm{W}_c = \big(\bm{K}^{\top}\bm{W}^{-1}\bm{K}\big)^{-1}$, which is well defined since $\bm{K}^{\top}\bm{W}^{-1}\bm{K}$ has full rank.

\subsection*{Structural form (\ref{eq:occ})}
\noindent  Write the coherent forecast as $\widetilde{\bm{z}}^{c} = \bm{S}\bm{G}\widehat{\bm{z}}$, which satisfies $\bm{C}\widetilde{\bm{z}}^{c} = \bm{0}$ for any $\bm{G}$. Unbiasedness, $\mathrm{E}[\widetilde{\bm{z}}^{c}] = \bm{S}\bm{y}$ for every $\bm{y}$, requires $\bm{S}\bm{G}\bm{K}\bm{S} = \bm{S}$, that is $\bm{G}\bm{S}_Q = \bm{I}_m$ with $\bm{S}_Q = \bm{K}\bm{S}$, since $\bm{S}$ has full column rank. The combined forecast error is $\widetilde{\bm{z}}^{c} - \bm{z} = \bm{S}\bm{G}(\widehat{\bm{z}} - \bm{K}\bm{z}) = \bm{S}\bm{G}\widehat{\bm{\varepsilon}}$, with covariance $\bm{S}\bm{G}\bm{W}\bm{G}^{\top}\bm{S}^{\top}$. Minimizing its trace under the constraint, consider the Lagrangian
\begin{equation*}
\mathcal{L}(\bm{G},\bm{\Lambda}) = \mathrm{tr}\big(\bm{S}\bm{G}\bm{W}\bm{G}^{\top}\bm{S}^{\top}\big) - 2\mathrm{tr}\big[\bm{\Lambda}^{\top}\big(\bm{G}\bm{S}_Q - \bm{I}_m\big)\big], \qquad \bm{\Lambda} \in \mathbb{R}^{m \times m},
\end{equation*}
whose first-order conditions are $\bm{W}\bm{G}^{\top}\bm{S}^{\top}\bm{S} = \bm{S}_Q\bm{\Lambda}$ and $\bm{G}\bm{S}_Q = \bm{I}_m$. Pre-multiplying the first condition by $\bm{S}_Q^{\top}\bm{W}^{-1} = \bm{S}^{\top}\bm{K}^{\top}\bm{W}^{-1}$ and using $\bm{K}^{\top}\bm{W}^{-1}\bm{K} = \bm{W}_c^{-1}$ together with $\bm{S}_Q^{\top}\bm{G}^{\top}\bm{S}^{\top}\bm{S} = \bm{S}^{\top}\bm{S}$ from the constraint, gives $\bm{\Lambda} = \big(\bm{S}^{\top}\bm{W}_c^{-1}\bm{S}\big)^{-1}\bm{S}^{\top}\bm{S}$. Substituting back, the factor $\bm{S}^{\top}\bm{S}$ cancels and
\begin{equation*}
\bm{G} = \big(\bm{S}^{\top}\bm{W}_c^{-1}\bm{S}\big)^{-1}\bm{S}^{\top}\bm{K}^{\top}\bm{W}^{-1} = \big(\bm{S}_Q^{\top}\bm{W}^{-1}\bm{S}_Q\big)^{-1}\bm{S}_Q^{\top}\bm{W}^{-1},
\end{equation*}
so that $\widetilde{\bm{z}}^{c} = \bm{S}\bm{G}\widehat{\bm{z}}$ is the structural form~\eqref{eq:occ}.

\subsection*{Zero-constrained form (\ref{eq:occproj})}
\noindent Following the zero-constrained approach, $\widetilde{\bm{z}}^{c}$ solves the linearly constrained quadratic program
\begin{equation*}
\widetilde{\bm{z}}^{c} = \arg\min_{\bm{z}}\ \big(\widehat{\bm{z}} - \bm{K}\bm{z}\big)^{\top}\bm{W}^{-1}\big(\widehat{\bm{z}} - \bm{K}\bm{z}\big) \quad \text{s.t.}\quad \bm{C}\bm{z} = \bm{0}.
\end{equation*}
With multipliers $\bm{\lambda}$, the Lagrangian $\big(\widehat{\bm{z}} - \bm{K}\bm{z}\big)^{\top}\bm{W}^{-1}\big(\widehat{\bm{z}} - \bm{K}\bm{z}\big) + 2\bm{\lambda}^{\top}\bm{C}\bm{z}$ yields the first-order system
\begin{equation*}
\begin{bmatrix} \bm{K}^{\top}\bm{W}^{-1}\bm{K} & \bm{C}^{\top} \\ \bm{C} & \bm{0} \end{bmatrix}
\begin{bmatrix} \widetilde{\bm{z}}^{c} \\ \widetilde{\bm{\lambda}} \end{bmatrix}
=
\begin{bmatrix} \bm{K}^{\top}\bm{W}^{-1}\widehat{\bm{z}} \\ \bm{0} \end{bmatrix}.
\end{equation*}
Solving the block system gives
\begin{equation*}
\widetilde{\bm{z}}^{c} = \Big[\bm{W}_c - \bm{W}_c\bm{C}^{\top}\big(\bm{C}\bm{W}_c\bm{C}^{\top}\big)^{-1}\bm{C}\bm{W}_c\Big]\bm{K}^{\top}\bm{W}^{-1}\widehat{\bm{z}} = \Big[\bm{I}_n - \bm{W}_c\bm{C}^{\top}\big(\bm{C}\bm{W}_c\bm{C}^{\top}\big)^{-1}\bm{C}\Big]\bm{W}_c\bm{K}^{\top}\bm{W}^{-1}\widehat{\bm{z}},
\end{equation*}
which is the zero-constrained form~\eqref{eq:occproj}. Setting $\bm{M} = \bm{I}_n - \bm{W}_c\bm{C}^{\top}\big(\bm{C}\bm{W}_c\bm{C}^{\top}\big)^{-1}\bm{C}$ and $\bm{\Omega} = \bm{W}^{-1}\bm{K}\bm{W}_c$, it reads $\widetilde{\bm{z}}^{c} = \bm{M}\bm{\Omega}^{\top}\widehat{\bm{z}}$, the temporal restatement of Theorem~1 of \citet{Girolimetto2024-qe}.

\subsection*{Equivalence of the two forms}
\noindent The zero-constrained model underlying the second derivation is $\widehat{\bm{z}} = \bm{K}\bm{z} + \widehat{\bm{\varepsilon}}$, subject to $\bm{C}\bm{z} = \bm{0}$. The constraint is equivalent to $\bm{z} = \bm{S}\bm{y}$ and, since $\bm{C}\bm{S} = \bm{0}$, substituting it removes the restriction and produces the unrestricted model $\widehat{\bm{z}} = \bm{K}\bm{S}\bm{y} + \widehat{\bm{\varepsilon}} = \bm{S}_Q\bm{y} + \widehat{\bm{\varepsilon}}$, which is exactly the structural model of the first derivation. The minimum mean square error estimate of $\bm{z}$ is therefore the same under the two parametrizations.
%
%
%
%

\section{MCB test at all temporal levels} \label{app:mcb}

This section extends the MCB Nemenyi test reported in \autoref{fig:mcb} for the hourly and daily levels to all eight temporal levels of the hierarchy. Figures~\ref{fig:oa-mcb-ae} and~\ref{fig:oa-mcb-se} display the average ranks computed on the absolute and squared errors, respectively, comparing the base, reconciled and coherent combined forecasts in the German and Spanish markets.

\begin{figure}[H]
\centering
\includegraphics[width = \linewidth]{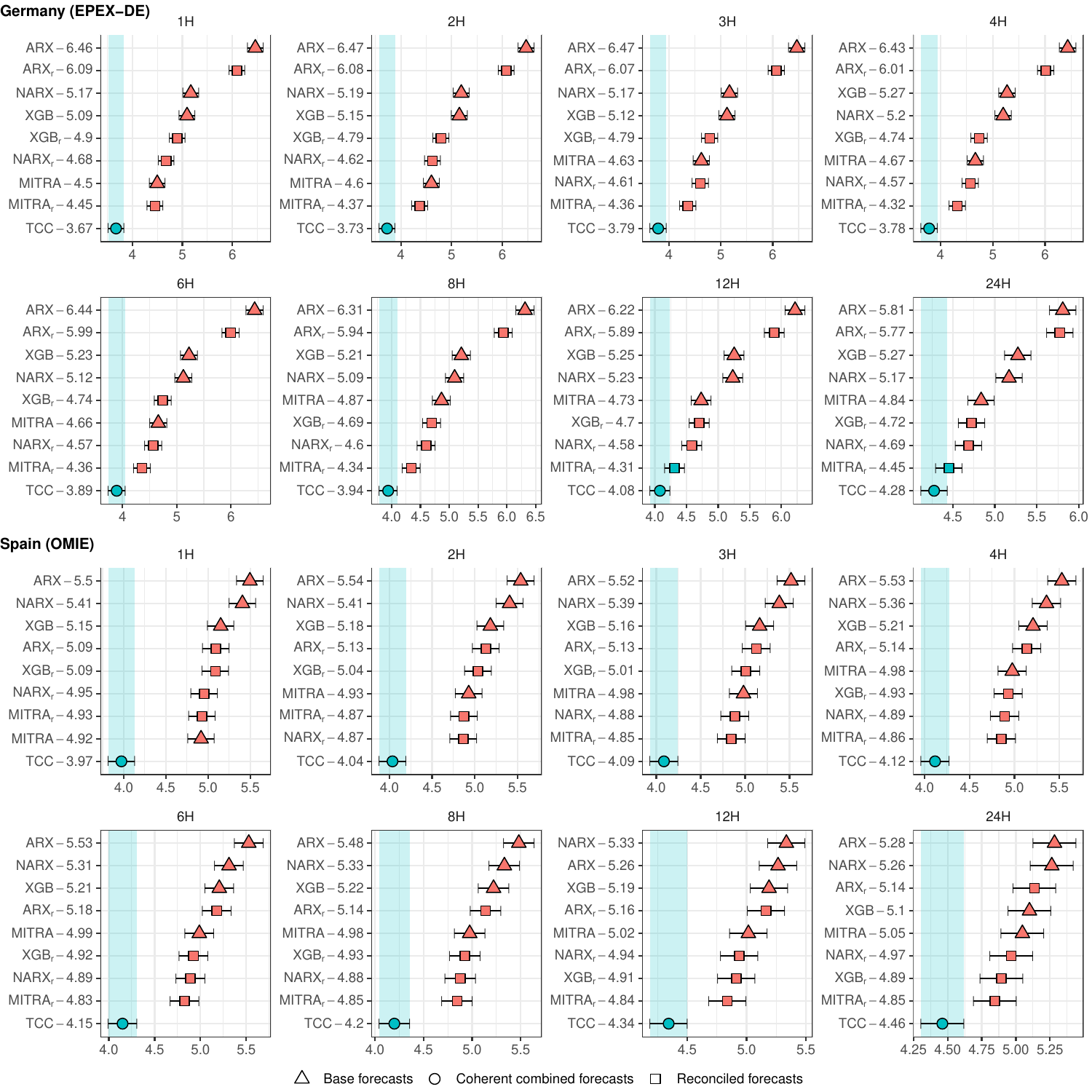}
\caption{Average ranks and MCB Nemenyi test intervals, computed on the absolute errors over the test days, at all temporal levels for the German (top) and Spanish (bottom) markets. TCC denotes the temporal coherent combination ($shr_{be}$ estimator) and the subscript $r$ the reconciled forecasts ($shr$ estimator). The shaded area marks the confidence region of the best-ranked approach: approaches whose intervals do not overlap it perform significantly worse.}
\label{fig:oa-mcb-ae}
\end{figure}

\begin{figure}[H]
\centering
\includegraphics[width = \linewidth]{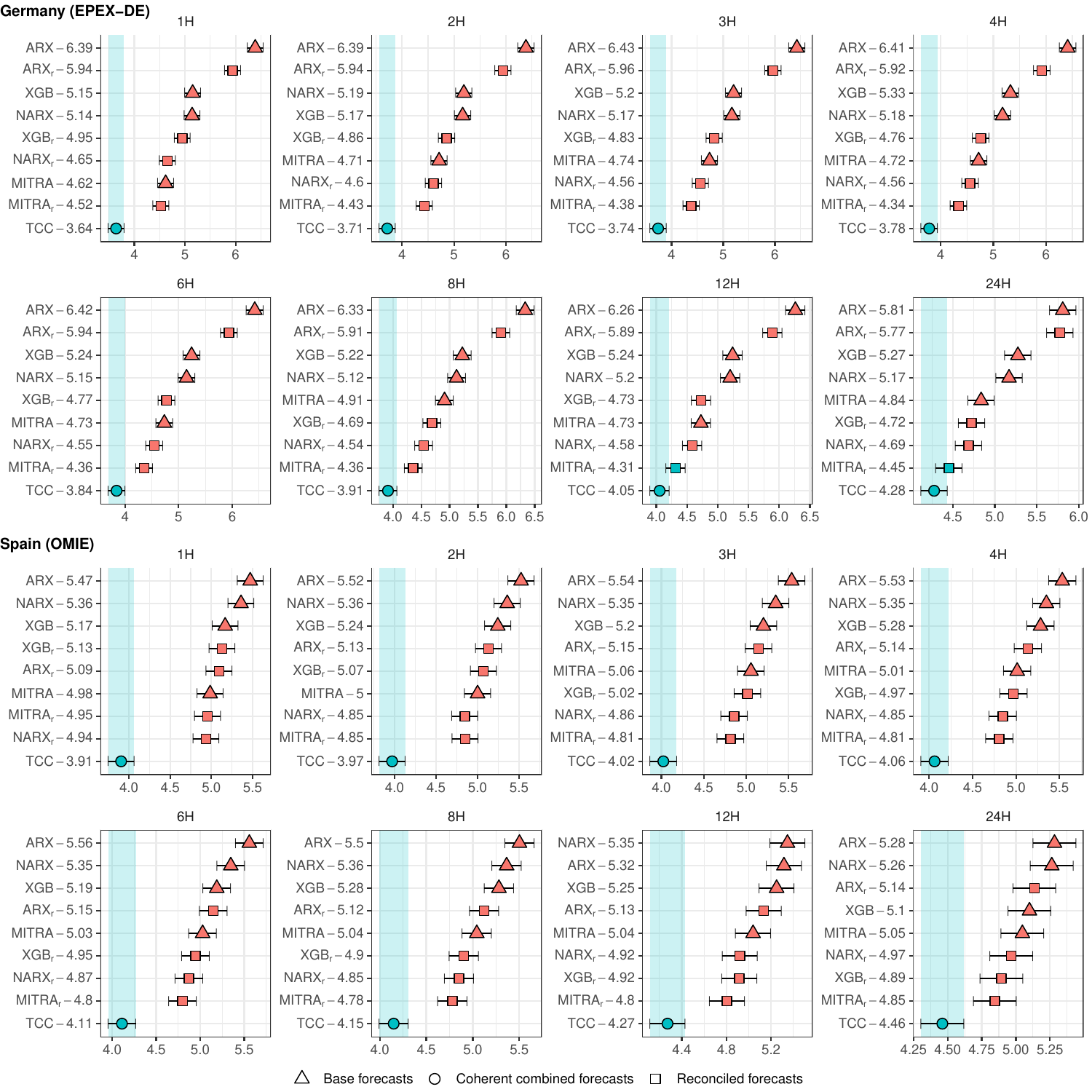}
\caption{Average ranks and MCB Nemenyi test intervals, computed on the squared errors over the test days, at all temporal levels for the German (top) and Spanish (bottom) markets. Notation as in Figure~\ref{fig:oa-mcb-ae}.}
\label{fig:oa-mcb-se}
\end{figure}

\section{Robustness analysis: RMSE results} \label{app:rmse}

This section collects the RMSE counterparts of the MAE tables presented in Section~\ref{sec:robust}. Table~\ref{tab:oa-covariance} compares the correlation structures, Table~\ref{tab:oa-shrinkage} the shrinkage techniques and Table~\ref{tab:oa-seq} the sequential approaches, leading to the same conclusions drawn from the MAE results.

\begin{table}[H]
\centering
\linespread{1}
\small
\setlength{\tabcolsep}{3pt}
\caption{RMSE at each temporal level for the base forecasts, the single-expert reconciled forecasts with the $wls$ and $shr$ estimators, and the temporal coherent combination with the $wls$, $shr$, $shr_{bn}$ and $shr_{be}$ estimators. RMSE counterpart of the corresponding MAE table in Section~\ref{sec:robust}. For each market and temporal level, bold and italics denote the best and the second-best approach, respectively.}
\label{tab:oa-covariance}
\resizebox{\linewidth}{!}{
\input{./tables/tab_covariance_rmse.tex}}
\end{table}

\begin{table}[H]
\centering
\linespread{1}
\small
\setlength{\tabcolsep}{3pt}
\caption{RMSE of the temporal coherent combination for the linear ($shr$) and nonlinear ($lis$, $qis$, $gis$) shrinkage techniques, each applied to the full matrix (no subscript), to the by-node blocks (subscript $bn$) and to the by-expert blocks (subscript $be$). RMSE counterpart of the corresponding MAE table in Section~\ref{sec:robust}. For each market and temporal level, bold and italics denote the best and the second-best estimator, respectively.}
\label{tab:oa-shrinkage}
\resizebox{\linewidth}{!}{\input{./tables/tab_shr_rmse.tex}}
\end{table}

\begin{table}[H]
\centering
\linespread{1}
\small
\setlength{\tabcolsep}{3pt}
\caption{RMSE at each temporal level of the sequential approaches, in which the equal-weighted average of the four experts ($sa$) is reconciled as a single expert with the $wls$ and $shr$ estimators, compared with the best single-expert reconciliation (MITRA base forecasts, $shr$ estimator) and the best temporal coherent combination ($shr_{be}$ estimator). RMSE counterpart of the corresponding MAE table in Section~\ref{sec:robust}. For each market and temporal level, bold and italics denote the best and the second-best approach, respectively.}
\label{tab:oa-seq}
\resizebox{\linewidth}{!}{\input{./tables/tab_seq_rmse.tex}}
{\scriptsize \vskip0.25em\textbf{$^\ast$} Simple average combination is not coherent.}
\end{table}

\phantomsection\addcontentsline{toc}{section}{References}
\bibliographystyle{apalike3link}
\bibliography{biblio.bib}

\end{document}

%% file: tables/shrbe_tab_full_dm.tex
\begin{tabular}[t]{cc|>{}cccccccccc|>{}cccccccccc}
\toprule
\multicolumn{2}{c}{\textbf{ }} & \multicolumn{10}{c}{\textbf{Germany (EPEX-DE)}} & \multicolumn{10}{c}{\textbf{Spain (OMIE)}} \\
\cmidrule(l{3pt}r{3pt}){3-12} \cmidrule(l{3pt}r{3pt}){13-22}
\textbf{Level} & \textbf{Model} & MAE & $\text{MAE}_{r}$ & \% & $\text{MAE}_{c}$ & \% & RMSE & $\text{RMSE}_{r}$ & \% & $\text{RMSE}_{c}$ & \% & MAE & $\text{MAE}_{r}$ & \% & $\text{MAE}_{c}$ & \% & RMSE & $\text{RMSE}_{r}$ & \% & $\text{RMSE}_{c}$ & \%\\
\midrule
 & ARX & \;26.93$\;$ & \;26.08$\;$ & \;\cellcolor[HTML]{EEF8EE}{3.1}$\;$ &  & \;\cellcolor[HTML]{80C680}{\textbf{23.9}}$\;$ & \;39.83$\;$ & \;38.56$\;$ & \;\cellcolor[HTML]{EEF7EE}{3.2}$\;$ &  & \;\cellcolor[HTML]{99D199}{\textbf{19.3}}$\;$ & \;16.72$\;$ & \;16.52$\;$ & \;\cellcolor[HTML]{F9FCF9}{1.1}$\;$ &  & \;\cellcolor[HTML]{CAE7CA}{\textbf{10.1}}$\;$ & \;24.40$\;$ & \;24.28$\;$ & \;\cellcolor[HTML]{FCFEFC}{0.5}$\,^{1}$ &  & \;\cellcolor[HTML]{D3EBD3}{\textbf{8.4}}$\;$\\

 & NARX & \;23.06$\;$ & \;22.27$\;$ & \;\cellcolor[HTML]{EDF7ED}{3.4}$\;$ &  & \;\cellcolor[HTML]{C4E5C4}{\textbf{11.1}}$\;$ & \;35.62$\;$ & \;34.56$\;$ & \;\cellcolor[HTML]{EFF8EF}{3.0}$\;$ &  & \;\cellcolor[HTML]{CBE8CB}{\textbf{9.8}}$\;$ & \;16.57$\;$ & \;16.34$\;$ & \;\cellcolor[HTML]{F7FCF7}{1.4}$\,^{1}$ &  & \;\cellcolor[HTML]{CEE9CE}{\textbf{9.3}}$\;$ & \;24.07$\;$ & \;24.02$\;$ & \;\cellcolor[HTML]{FEFEFE}{0.2}$\,^{1}$ &  & \;\cellcolor[HTML]{D9EED9}{\textbf{7.1}}$\;$\\

 & XGB & \;23.33$\;$ & \;22.36$\;$ & \;\cellcolor[HTML]{E9F5E9}{4.2}$\;$ &  & \;\cellcolor[HTML]{BFE2BF}{\textbf{12.2}}$\;$ & \;37.05$\;$ & \;35.13$\;$ & \;\cellcolor[HTML]{E4F3E4}{5.2}$\;$ &  & \;\cellcolor[HTML]{B9DFB9}{\textbf{13.2}}$\;$ & \;16.54$\;$ & \;16.00$\;$ & \;\cellcolor[HTML]{EEF7EE}{3.2}$\;$ &  & \;\cellcolor[HTML]{CFE9CF}{\textbf{9.1}}$\;$ & \;24.57$\;$ & \;23.70$\;$ & \;\cellcolor[HTML]{ECF7EC}{3.5}$\;$ &  & \;\cellcolor[HTML]{CFEACF}{\textbf{9.0}}$\;$\\

\multirow{-4}{*}[\normalbaselineskip]{\centering\arraybackslash 1H} & MITRA & \;21.37$\;$ & \;21.08$\;$ & \;\cellcolor[HTML]{F8FCF8}{1.4}$\;$ & \multirow{-4}{*}[\normalbaselineskip]{\centering\arraybackslash \;20.49$\;$} & \;\cellcolor[HTML]{E9F5E9}{\textbf{4.1}}$\;$ & \;33.31$\;$ & \;32.84$\;$ & \;\cellcolor[HTML]{F7FCF7}{1.4}$\,^{1}$ & \multirow{-4}{*}[\normalbaselineskip]{\centering\arraybackslash \;32.14$\;$} & \;\cellcolor[HTML]{ECF7EC}{\textbf{3.5}}$\,^{2}$ & \;15.92$\;$ & \;15.88$\;$ & \;\cellcolor[HTML]{FEFEFE}{0.3}$\,^{1}$ & \multirow{-4}{*}[\normalbaselineskip]{\centering\arraybackslash \;15.04$\;$} & \;\cellcolor[HTML]{E1F2E1}{\textbf{5.6}}$\;$ & \;23.61$\;$ & \;23.25$\;$ & \;\cellcolor[HTML]{F7FBF7}{1.5}$\,^{1}$ & \multirow{-4}{*}[\normalbaselineskip]{\centering\arraybackslash \;22.36$\;$} & \;\cellcolor[HTML]{E3F2E3}{\textbf{5.3}}$\;$\\

\midrule
 & ARX & \;26.44$\;$ & \;25.60$\;$ & \;\cellcolor[HTML]{EEF7EE}{3.2}$\;$ &  & \;\cellcolor[HTML]{7DC47D}{\textbf{24.6}}$\;$ & \;38.95$\;$ & \;37.66$\;$ & \;\cellcolor[HTML]{EEF7EE}{3.3}$\;$ &  & \;\cellcolor[HTML]{94CF94}{\textbf{20.1}}$\;$ & \;16.41$\;$ & \;16.20$\;$ & \;\cellcolor[HTML]{F8FCF8}{1.2}$\;$ &  & \;\cellcolor[HTML]{C7E6C7}{\textbf{10.5}}$\;$ & \;23.96$\;$ & \;23.82$\;$ & \;\cellcolor[HTML]{FCFEFC}{0.6}$\,^{1}$ &  & \;\cellcolor[HTML]{D1EAD1}{\textbf{8.8}}$\;$\\

 & NARX & \;22.64$\;$ & \;21.59$\;$ & \;\cellcolor[HTML]{E7F4E7}{4.6}$\;$ &  & \;\cellcolor[HTML]{C0E3C0}{\textbf{11.9}}$\;$ & \;34.84$\;$ & \;33.37$\;$ & \;\cellcolor[HTML]{E9F5E9}{4.2}$\;$ &  & \;\cellcolor[HTML]{C6E5C6}{\textbf{10.7}}$\;$ & \;16.29$\;$ & \;15.89$\;$ & \;\cellcolor[HTML]{F2F9F2}{2.4}$\;$ &  & \;\cellcolor[HTML]{CBE7CB}{\textbf{9.9}}$\;$ & \;23.71$\;$ & \;23.39$\;$ & \;\cellcolor[HTML]{F8FCF8}{1.4}$\,^{1}$ &  & \;\cellcolor[HTML]{D6ECD6}{\textbf{7.8}}$\;$\\

 & XGB & \;22.74$\;$ & \;21.72$\;$ & \;\cellcolor[HTML]{E7F4E7}{4.5}$\;$ &  & \;\cellcolor[HTML]{BEE2BE}{\textbf{12.4}}$\;$ & \;35.87$\;$ & \;33.99$\;$ & \;\cellcolor[HTML]{E3F2E3}{5.3}$\;$ &  & \;\cellcolor[HTML]{B8DFB8}{\textbf{13.3}}$\;$ & \;16.19$\;$ & \;15.59$\;$ & \;\cellcolor[HTML]{EBF6EB}{3.7}$\;$ &  & \;\cellcolor[HTML]{CEE9CE}{\textbf{9.3}}$\;$ & \;24.11$\;$ & \;23.13$\;$ & \;\cellcolor[HTML]{E9F5E9}{4.1}$\;$ &  & \;\cellcolor[HTML]{CEE9CE}{\textbf{9.3}}$\;$\\

\multirow{-4}{*}[\normalbaselineskip]{\centering\arraybackslash 2H} & MITRA & \;20.92$\;$ & \;20.34$\;$ & \;\cellcolor[HTML]{F0F8F0}{2.8}$\;$ & \multirow{-4}{*}[\normalbaselineskip]{\centering\arraybackslash \;19.94$\;$} & \;\cellcolor[HTML]{E6F4E6}{\textbf{4.7}}$\,^{2}$ & \;32.43$\;$ & \;31.52$\;$ & \;\cellcolor[HTML]{F0F8F0}{2.8}$\;$ & \multirow{-4}{*}[\normalbaselineskip]{\centering\arraybackslash \;31.10$\;$} & \;\cellcolor[HTML]{E9F5E9}{\textbf{4.1}}$\,^{2}$ & \;15.56$\;$ & \;15.43$\;$ & \;\cellcolor[HTML]{FBFDFB}{0.8}$\,^{1}$ & \multirow{-4}{*}[\normalbaselineskip]{\centering\arraybackslash \;14.68$\;$} & \;\cellcolor[HTML]{E1F2E1}{\textbf{5.7}}$\;$ & \;23.17$\;$ & \;22.59$\;$ & \;\cellcolor[HTML]{F2F9F2}{2.5}$\,^{1}$ & \multirow{-4}{*}[\normalbaselineskip]{\centering\arraybackslash \;21.86$\;$} & \;\cellcolor[HTML]{E1F2E1}{\textbf{5.6}}$\;$\\

\midrule
 & ARX & \;26.11$\;$ & \;25.26$\;$ & \;\cellcolor[HTML]{EEF7EE}{3.2}$\;$ &  & \;\cellcolor[HTML]{7AC37A}{\textbf{25.1}}$\;$ & \;38.28$\;$ & \;36.98$\;$ & \;\cellcolor[HTML]{EDF7ED}{3.4}$\;$ &  & \;\cellcolor[HTML]{91CD91}{\textbf{20.8}}$\;$ & \;16.16$\;$ & \;15.96$\;$ & \;\cellcolor[HTML]{F8FCF8}{1.3}$\;$ &  & \;\cellcolor[HTML]{C6E5C6}{\textbf{10.8}}$\;$ & \;23.59$\;$ & \;23.43$\;$ & \;\cellcolor[HTML]{FBFDFB}{0.7}$\;$ &  & \;\cellcolor[HTML]{CFEACF}{\textbf{9.0}}$\;$\\

 & NARX & \;22.12$\;$ & \;21.14$\;$ & \;\cellcolor[HTML]{E7F4E7}{4.4}$\;$ &  & \;\cellcolor[HTML]{C1E3C1}{\textbf{11.7}}$\;$ & \;33.89$\;$ & \;32.55$\;$ & \;\cellcolor[HTML]{EAF6EA}{4.0}$\;$ &  & \;\cellcolor[HTML]{C7E6C7}{\textbf{10.5}}$\;$ & \;16.01$\;$ & \;15.63$\;$ & \;\cellcolor[HTML]{F3F9F3}{2.3}$\;$ &  & \;\cellcolor[HTML]{CAE7CA}{\textbf{9.9}}$\;$ & \;23.44$\;$ & \;22.96$\;$ & \;\cellcolor[HTML]{F4FAF4}{2.1}$\,^{1}$ &  & \;\cellcolor[HTML]{D2EBD2}{\textbf{8.4}}$\;$\\

 & XGB & \;22.26$\;$ & \;21.30$\;$ & \;\cellcolor[HTML]{E8F5E8}{4.3}$\;$ &  & \;\cellcolor[HTML]{BEE2BE}{\textbf{12.2}}$\;$ & \;35.10$\;$ & \;33.20$\;$ & \;\cellcolor[HTML]{E2F2E2}{5.4}$\;$ &  & \;\cellcolor[HTML]{B7DFB7}{\textbf{13.6}}$\;$ & \;15.93$\;$ & \;15.29$\;$ & \;\cellcolor[HTML]{EAF5EA}{4.0}$\;$ &  & \;\cellcolor[HTML]{CDE8CD}{\textbf{9.5}}$\;$ & \;23.75$\;$ & \;22.71$\;$ & \;\cellcolor[HTML]{E8F5E8}{4.4}$\;$ &  & \;\cellcolor[HTML]{CCE8CC}{\textbf{9.6}}$\;$\\

\multirow{-4}{*}[\normalbaselineskip]{\centering\arraybackslash 3H} & MITRA & \;20.61$\;$ & \;19.86$\;$ & \;\cellcolor[HTML]{ECF6EC}{3.7}$\;$ & \multirow{-4}{*}[\normalbaselineskip]{\centering\arraybackslash \;19.54$\;$} & \;\cellcolor[HTML]{E3F3E3}{\textbf{5.2}}$\,^{2}$ & \;31.87$\;$ & \;30.61$\;$ & \;\cellcolor[HTML]{EAF6EA}{3.9}$\;$ & \multirow{-4}{*}[\normalbaselineskip]{\centering\arraybackslash \;30.33$\;$} & \;\cellcolor[HTML]{E6F4E6}{\textbf{4.8}}$\,^{2}$ & \;15.31$\;$ & \;15.14$\;$ & \;\cellcolor[HTML]{F9FCF9}{1.1}$\,^{1}$ & \multirow{-4}{*}[\normalbaselineskip]{\centering\arraybackslash \;14.42$\;$} & \;\cellcolor[HTML]{E0F1E0}{\textbf{5.8}}$\;$ & \;22.72$\;$ & \;22.14$\;$ & \;\cellcolor[HTML]{F1F9F1}{2.6}$\;$ & \multirow{-4}{*}[\normalbaselineskip]{\centering\arraybackslash \;21.47$\;$} & \;\cellcolor[HTML]{E2F2E2}{\textbf{5.5}}$\;$\\

\midrule
 & ARX & \;25.65$\;$ & \;24.76$\;$ & \;\cellcolor[HTML]{EDF7ED}{3.4}$\;$ &  & \;\cellcolor[HTML]{77C277}{\textbf{25.7}}$\;$ & \;37.50$\;$ & \;36.13$\;$ & \;\cellcolor[HTML]{ECF6EC}{3.6}$\;$ &  & \;\cellcolor[HTML]{8ECC8E}{\textbf{21.3}}$\;$ & \;15.88$\;$ & \;15.67$\;$ & \;\cellcolor[HTML]{F8FCF8}{1.3}$\;$ &  & \;\cellcolor[HTML]{C5E5C5}{\textbf{11.0}}$\;$ & \;23.25$\;$ & \;23.07$\;$ & \;\cellcolor[HTML]{FBFDFB}{0.8}$\;$ &  & \;\cellcolor[HTML]{CEE9CE}{\textbf{9.2}}$\;$\\

 & NARX & \;21.75$\;$ & \;20.60$\;$ & \;\cellcolor[HTML]{E3F2E3}{5.3}$\;$ &  & \;\cellcolor[HTML]{BDE1BD}{\textbf{12.4}}$\;$ & \;33.27$\;$ & \;31.67$\;$ & \;\cellcolor[HTML]{E5F4E5}{4.8}$\;$ &  & \;\cellcolor[HTML]{C3E4C3}{\textbf{11.3}}$\;$ & \;15.68$\;$ & \;15.36$\;$ & \;\cellcolor[HTML]{F4FAF4}{2.1}$\,^{1}$ &  & \;\cellcolor[HTML]{CBE7CB}{\textbf{9.9}}$\;$ & \;22.96$\;$ & \;22.57$\;$ & \;\cellcolor[HTML]{F6FBF6}{1.7}$\,^{1}$ &  & \;\cellcolor[HTML]{D4ECD4}{\textbf{8.1}}$\;$\\

 & XGB & \;22.01$\;$ & \;20.78$\;$ & \;\cellcolor[HTML]{E1F2E1}{5.6}$\;$ &  & \;\cellcolor[HTML]{B8DFB8}{\textbf{13.4}}$\;$ & \;34.24$\;$ & \;32.29$\;$ & \;\cellcolor[HTML]{E1F1E1}{5.7}$\;$ &  & \;\cellcolor[HTML]{B6DEB6}{\textbf{13.8}}$\;$ & \;15.78$\;$ & \;14.97$\;$ & \;\cellcolor[HTML]{E4F3E4}{5.2}$\;$ &  & \;\cellcolor[HTML]{C8E6C8}{\textbf{10.5}}$\;$ & \;23.60$\;$ & \;22.29$\;$ & \;\cellcolor[HTML]{E2F2E2}{5.5}$\;$ &  & \;\cellcolor[HTML]{C7E6C7}{\textbf{10.5}}$\;$\\

\multirow{-4}{*}[\normalbaselineskip]{\centering\arraybackslash 4H} & MITRA & \;20.13$\;$ & \;19.34$\;$ & \;\cellcolor[HTML]{EAF6EA}{3.9}$\;$ & \multirow{-4}{*}[\normalbaselineskip]{\centering\arraybackslash \;19.06$\;$} & \;\cellcolor[HTML]{E3F2E3}{\textbf{5.3}}$\,^{2}$ & \;30.86$\;$ & \;29.70$\;$ & \;\cellcolor[HTML]{EBF6EB}{3.8}$\;$ & \multirow{-4}{*}[\normalbaselineskip]{\centering\arraybackslash \;29.51$\;$} & \;\cellcolor[HTML]{E8F5E8}{\textbf{4.4}}$\,^{2}$ & \;15.01$\;$ & \;14.84$\;$ & \;\cellcolor[HTML]{F9FCF9}{1.1}$\,^{1}$ & \multirow{-4}{*}[\normalbaselineskip]{\centering\arraybackslash \;14.13$\;$} & \;\cellcolor[HTML]{E0F1E0}{\textbf{5.8}}$\;$ & \;22.19$\;$ & \;21.75$\;$ & \;\cellcolor[HTML]{F4FAF4}{2.0}$\,^{1}$ & \multirow{-4}{*}[\normalbaselineskip]{\centering\arraybackslash \;21.11$\;$} & \;\cellcolor[HTML]{E5F3E5}{\textbf{4.9}}$\;$\\

\midrule
 & ARX & \;25.14$\;$ & \;24.28$\;$ & \;\cellcolor[HTML]{EDF7ED}{3.4}$\;$ &  & \;\cellcolor[HTML]{74C174}{\textbf{26.2}}$\;$ & \;36.42$\;$ & \;35.12$\;$ & \;\cellcolor[HTML]{ECF7EC}{3.6}$\;$ &  & \;\cellcolor[HTML]{8ACA8A}{\textbf{22.2}}$\;$ & \;15.60$\;$ & \;15.39$\;$ & \;\cellcolor[HTML]{F8FCF8}{1.3}$\;$ &  & \;\cellcolor[HTML]{C1E3C1}{\textbf{11.7}}$\;$ & \;22.77$\;$ & \;22.55$\;$ & \;\cellcolor[HTML]{FAFDFA}{1.0}$\;$ &  & \;\cellcolor[HTML]{CDE8CD}{\textbf{9.5}}$\;$\\

 & NARX & \;21.08$\;$ & \;19.94$\;$ & \;\cellcolor[HTML]{E2F2E2}{5.4}$\;$ &  & \;\cellcolor[HTML]{C0E2C0}{\textbf{12.0}}$\;$ & \;32.00$\;$ & \;30.42$\;$ & \;\cellcolor[HTML]{E5F3E5}{4.9}$\;$ &  & \;\cellcolor[HTML]{C2E4C2}{\textbf{11.4}}$\;$ & \;15.27$\;$ & \;14.95$\;$ & \;\cellcolor[HTML]{F4FAF4}{2.1}$\,^{1}$ &  & \;\cellcolor[HTML]{CBE8CB}{\textbf{9.8}}$\;$ & \;22.41$\;$ & \;22.02$\;$ & \;\cellcolor[HTML]{F6FBF6}{1.8}$\,^{1}$ &  & \;\cellcolor[HTML]{D4ECD4}{\textbf{8.1}}$\;$\\

 & XGB & \;21.52$\;$ & \;20.20$\;$ & \;\cellcolor[HTML]{DEF0DE}{6.2}$\;$ &  & \;\cellcolor[HTML]{B6DEB6}{\textbf{13.7}}$\;$ & \;33.38$\;$ & \;31.15$\;$ & \;\cellcolor[HTML]{DCEFDC}{6.7}$\;$ &  & \;\cellcolor[HTML]{AFDBAF}{\textbf{15.1}}$\;$ & \;15.41$\;$ & \;14.61$\;$ & \;\cellcolor[HTML]{E4F3E4}{5.2}$\;$ &  & \;\cellcolor[HTML]{C7E6C7}{\textbf{10.6}}$\;$ & \;22.93$\;$ & \;21.77$\;$ & \;\cellcolor[HTML]{E4F3E4}{5.0}$\;$ &  & \;\cellcolor[HTML]{C9E7C9}{\textbf{10.1}}$\;$\\

\multirow{-4}{*}[\normalbaselineskip]{\centering\arraybackslash 6H} & MITRA & \;19.51$\;$ & \;18.71$\;$ & \;\cellcolor[HTML]{E9F5E9}{4.1}$\;$ & \multirow{-4}{*}[\normalbaselineskip]{\centering\arraybackslash \;18.56$\;$} & \;\cellcolor[HTML]{E5F3E5}{\textbf{4.8}}$\,^{2}$ & \;29.73$\;$ & \;28.38$\;$ & \;\cellcolor[HTML]{E7F4E7}{4.6}$\;$ & \multirow{-4}{*}[\normalbaselineskip]{\centering\arraybackslash \;28.34$\;$} & \;\cellcolor[HTML]{E6F4E6}{\textbf{4.7}}$\,^{2}$ & \;14.77$\;$ & \;14.44$\;$ & \;\cellcolor[HTML]{F3FAF3}{2.2}$\;$ & \multirow{-4}{*}[\normalbaselineskip]{\centering\arraybackslash \;13.78$\;$} & \;\cellcolor[HTML]{DBEFDB}{\textbf{6.7}}$\;$ & \;21.87$\;$ & \;21.18$\;$ & \;\cellcolor[HTML]{EEF8EE}{3.1}$\;$ & \multirow{-4}{*}[\normalbaselineskip]{\centering\arraybackslash \;20.60$\;$} & \;\cellcolor[HTML]{E0F1E0}{\textbf{5.8}}$\;$\\

\midrule
 & ARX & \;24.37$\;$ & \;23.59$\;$ & \;\cellcolor[HTML]{EEF7EE}{3.2}$\;$ &  & \;\cellcolor[HTML]{76C176}{\textbf{25.8}}$\;$ & \;35.59$\;$ & \;34.15$\;$ & \;\cellcolor[HTML]{EAF5EA}{4.0}$\;$ &  & \;\cellcolor[HTML]{89CA89}{\textbf{22.3}}$\;$ & \;15.14$\;$ & \;14.95$\;$ & \;\cellcolor[HTML]{F8FCF8}{1.3}$\;$ &  & \;\cellcolor[HTML]{C2E4C2}{\textbf{11.4}}$\;$ & \;22.12$\;$ & \;21.90$\;$ & \;\cellcolor[HTML]{FAFDFA}{1.0}$\;$ &  & \;\cellcolor[HTML]{CCE8CC}{\textbf{9.6}}$\;$\\

 & NARX & \;20.54$\;$ & \;19.39$\;$ & \;\cellcolor[HTML]{E1F2E1}{5.6}$\;$ &  & \;\cellcolor[HTML]{C0E2C0}{\textbf{12.0}}$\;$ & \;31.20$\;$ & \;29.61$\;$ & \;\cellcolor[HTML]{E4F3E4}{5.1}$\;$ &  & \;\cellcolor[HTML]{C2E4C2}{\textbf{11.4}}$\;$ & \;15.00$\;$ & \;14.61$\;$ & \;\cellcolor[HTML]{F1F9F1}{2.6}$\;$ &  & \;\cellcolor[HTML]{C7E6C7}{\textbf{10.6}}$\;$ & \;21.92$\;$ & \;21.42$\;$ & \;\cellcolor[HTML]{F3FAF3}{2.3}$\,^{1}$ &  & \;\cellcolor[HTML]{D1EAD1}{\textbf{8.7}}$\;$\\

 & XGB & \;21.06$\;$ & \;19.68$\;$ & \;\cellcolor[HTML]{DCEFDC}{6.6}$\;$ &  & \;\cellcolor[HTML]{B4DDB4}{\textbf{14.2}}$\;$ & \;32.94$\;$ & \;30.34$\;$ & \;\cellcolor[HTML]{D5ECD5}{7.9}$\;$ &  & \;\cellcolor[HTML]{AAD9AA}{\textbf{16.1}}$\;$ & \;15.10$\;$ & \;14.20$\;$ & \;\cellcolor[HTML]{DFF1DF}{6.0}$\;$ &  & \;\cellcolor[HTML]{C3E4C3}{\textbf{11.2}}$\;$ & \;22.50$\;$ & \;21.12$\;$ & \;\cellcolor[HTML]{DEF0DE}{6.2}$\;$ &  & \;\cellcolor[HTML]{C4E5C4}{\textbf{11.1}}$\;$\\

\multirow{-4}{*}[\normalbaselineskip]{\centering\arraybackslash 8H} & MITRA & \;19.45$\;$ & \;18.18$\;$ & \;\cellcolor[HTML]{DDF0DD}{6.5}$\;$ & \multirow{-4}{*}[\normalbaselineskip]{\centering\arraybackslash \;18.08$\;$} & \;\cellcolor[HTML]{DAEEDA}{\textbf{7.0}}$\,^{2}$ & \;29.34$\;$ & \;27.63$\;$ & \;\cellcolor[HTML]{E0F1E0}{\textbf{5.8}}$\;$ & \multirow{-4}{*}[\normalbaselineskip]{\centering\arraybackslash \;27.64$\;$} & \;\cellcolor[HTML]{E0F1E0}{5.8}$\,^{2}$ & \;14.44$\;$ & \;14.06$\;$ & \;\cellcolor[HTML]{F1F9F1}{2.6}$\;$ & \multirow{-4}{*}[\normalbaselineskip]{\centering\arraybackslash \;13.41$\;$} & \;\cellcolor[HTML]{D9EED9}{\textbf{7.1}}$\;$ & \;21.45$\;$ & \;20.53$\;$ & \;\cellcolor[HTML]{E8F5E8}{4.3}$\,^{1}$ & \multirow{-4}{*}[\normalbaselineskip]{\centering\arraybackslash \;20.01$\;$} & \;\cellcolor[HTML]{DBEFDB}{\textbf{6.7}}$\,^{2}$\\

\midrule
 & ARX & \;22.81$\;$ & \;22.10$\;$ & \;\cellcolor[HTML]{EFF8EF}{3.1}$\;$ &  & \;\cellcolor[HTML]{6EBE6E}{\textbf{27.4}}$\;$ & \;32.68$\;$ & \;31.64$\;$ & \;\cellcolor[HTML]{EEF7EE}{3.2}$\;$ &  & \;\cellcolor[HTML]{83C783}{\textbf{23.4}}$\;$ & \;14.27$\;$ & \;14.19$\;$ & \;\cellcolor[HTML]{FCFEFC}{0.6}$\,^{1}$ &  & \;\cellcolor[HTML]{C7E6C7}{\textbf{10.5}}$\;$ & \;20.83$\;$ & \;20.68$\;$ & \;\cellcolor[HTML]{FBFDFB}{0.7}$\;$ &  & \;\cellcolor[HTML]{CFE9CF}{\textbf{9.0}}$\;$\\

 & NARX & \;19.37$\;$ & \;17.72$\;$ & \;\cellcolor[HTML]{D2EBD2}{8.5}$\;$ &  & \;\cellcolor[HTML]{B2DCB2}{\textbf{14.5}}$\;$ & \;28.95$\;$ & \;26.94$\;$ & \;\cellcolor[HTML]{DAEEDA}{7.0}$\;$ &  & \;\cellcolor[HTML]{B7DFB7}{\textbf{13.5}}$\;$ & \;14.25$\;$ & \;13.87$\;$ & \;\cellcolor[HTML]{F1F9F1}{2.7}$\,^{1}$ &  & \;\cellcolor[HTML]{C8E6C8}{\textbf{10.4}}$\;$ & \;20.64$\;$ & \;20.28$\;$ & \;\cellcolor[HTML]{F6FBF6}{1.7}$\,^{1}$ &  & \;\cellcolor[HTML]{D4ECD4}{\textbf{8.2}}$\;$\\

 & XGB & \;19.70$\;$ & \;17.98$\;$ & \;\cellcolor[HTML]{D1EAD1}{8.7}$\;$ &  & \;\cellcolor[HTML]{ABD9AB}{\textbf{15.9}}$\;$ & \;30.55$\;$ & \;27.61$\;$ & \;\cellcolor[HTML]{CCE8CC}{9.6}$\;$ &  & \;\cellcolor[HTML]{A0D4A0}{\textbf{18.0}}$\;$ & \;14.29$\;$ & \;13.50$\;$ & \;\cellcolor[HTML]{E2F2E2}{5.5}$\;$ &  & \;\cellcolor[HTML]{C7E6C7}{\textbf{10.7}}$\;$ & \;21.24$\;$ & \;20.01$\;$ & \;\cellcolor[HTML]{E0F1E0}{5.8}$\;$ &  & \;\cellcolor[HTML]{C6E5C6}{\textbf{10.8}}$\;$\\

\multirow{-4}{*}[\normalbaselineskip]{\centering\arraybackslash 12H} & MITRA & \;17.59$\;$ & \;16.50$\;$ & \;\cellcolor[HTML]{DEF0DE}{\textbf{6.2}}$\;$ & \multirow{-4}{*}[\normalbaselineskip]{\centering\arraybackslash \;16.56$\;$} & \;\cellcolor[HTML]{E0F1E0}{5.8}$\,^{2}$ & \;26.27$\;$ & \;24.69$\;$ & \;\cellcolor[HTML]{DFF1DF}{\textbf{6.0}}$\;$ & \multirow{-4}{*}[\normalbaselineskip]{\centering\arraybackslash \;25.04$\;$} & \;\cellcolor[HTML]{E6F4E6}{4.7}$\,^{1,2}$ & \;13.67$\;$ & \;13.34$\;$ & \;\cellcolor[HTML]{F2F9F2}{2.4}$\;$ & \multirow{-4}{*}[\normalbaselineskip]{\centering\arraybackslash \;12.77$\;$} & \;\cellcolor[HTML]{DCEFDC}{\textbf{6.6}}$\;$ & \;20.16$\;$ & \;19.44$\;$ & \;\cellcolor[HTML]{ECF6EC}{3.6}$\,^{1}$ & \multirow{-4}{*}[\normalbaselineskip]{\centering\arraybackslash \;18.95$\;$} & \;\cellcolor[HTML]{DFF1DF}{\textbf{6.0}}$\,^{2}$\\

\midrule
 & ARX & \;20.92$\;$ & \;20.52$\;$ & \;\cellcolor[HTML]{F5FAF5}{1.9}$\;$ &  & \;\cellcolor[HTML]{70BF70}{\textbf{26.9}}$\;$ & \;30.28$\;$ & \;29.15$\;$ & \;\cellcolor[HTML]{EBF6EB}{3.8}$\;$ &  & \;\cellcolor[HTML]{82C782}{\textbf{23.7}}$\;$ & \;13.57$\;$ & \;13.36$\;$ & \;\cellcolor[HTML]{F7FBF7}{1.5}$\;$ &  & \;\cellcolor[HTML]{BFE2BF}{\textbf{12.1}}$\;$ & \;19.50$\;$ & \;19.31$\;$ & \;\cellcolor[HTML]{FAFDFA}{1.0}$\,^{1}$ &  & \;\cellcolor[HTML]{CFE9CF}{\textbf{9.1}}$\;$\\

 & NARX & \;17.99$\;$ & \;16.41$\;$ & \;\cellcolor[HTML]{D0EAD0}{8.8}$\;$ &  & \;\cellcolor[HTML]{AFDBAF}{\textbf{15.1}}$\;$ & \;27.24$\;$ & \;24.90$\;$ & \;\cellcolor[HTML]{D1EBD1}{8.6}$\;$ &  & \;\cellcolor[HTML]{AFDBAF}{\textbf{15.1}}$\;$ & \;13.29$\;$ & \;13.00$\;$ & \;\cellcolor[HTML]{F4FAF4}{2.2}$\,^{1}$ &  & \;\cellcolor[HTML]{C8E6C8}{\textbf{10.3}}$\;$ & \;19.14$\;$ & \;19.01$\;$ & \;\cellcolor[HTML]{FBFDFB}{0.7}$\,^{1}$ &  & \;\cellcolor[HTML]{D8EDD8}{\textbf{7.4}}$\;$\\

 & XGB & \;18.80$\;$ & \;16.63$\;$ & \;\cellcolor[HTML]{C2E4C2}{11.5}$\;$ &  & \;\cellcolor[HTML]{9CD29C}{\textbf{18.7}}$\;$ & \;29.50$\;$ & \;25.60$\;$ & \;\cellcolor[HTML]{B9DFB9}{13.2}$\;$ &  & \;\cellcolor[HTML]{8CCB8C}{\textbf{21.6}}$\;$ & \;13.25$\;$ & \;12.57$\;$ & \;\cellcolor[HTML]{E4F3E4}{5.1}$\;$ &  & \;\cellcolor[HTML]{CAE7CA}{\textbf{10.0}}$\;$ & \;19.61$\;$ & \;18.66$\;$ & \;\cellcolor[HTML]{E5F3E5}{4.8}$\;$ &  & \;\cellcolor[HTML]{CCE8CC}{\textbf{9.6}}$\;$\\

\multirow{-4}{*}[\normalbaselineskip]{\centering\arraybackslash 24H} & MITRA & \;16.56$\;$ & \;15.15$\;$ & \;\cellcolor[HTML]{D2EBD2}{\textbf{8.5}}$\;$ & \multirow{-4}{*}[\normalbaselineskip]{\centering\arraybackslash \;15.28$\;$} & \;\cellcolor[HTML]{D6EDD6}{7.7}$\,^{2}$ & \;24.81$\;$ & \;22.56$\;$ & \;\cellcolor[HTML]{CFE9CF}{\textbf{9.1}}$\;$ & \multirow{-4}{*}[\normalbaselineskip]{\centering\arraybackslash \;23.12$\;$} & \;\cellcolor[HTML]{DBEFDB}{6.8}$\,^{2}$ & \;12.91$\;$ & \;12.43$\;$ & \;\cellcolor[HTML]{EBF6EB}{3.7}$\;$ & \multirow{-4}{*}[\normalbaselineskip]{\centering\arraybackslash \;11.92$\;$} & \;\cellcolor[HTML]{D7EDD7}{\textbf{7.6}}$\;$ & \;18.84$\;$ & \;18.17$\;$ & \;\cellcolor[HTML]{ECF7EC}{3.6}$\,^{1}$ & \multirow{-4}{*}[\normalbaselineskip]{\centering\arraybackslash \;17.73$\;$} & \;\cellcolor[HTML]{E0F1E0}{\textbf{5.9}}$\,^{2}$\\
\bottomrule
\end{tabular}

%% file: tables/tab_covariance_mae.tex
\begin{tabular}[t]{>{}c|ccccccc>{}c|cccccccc}
\toprule
\multicolumn{1}{c}{\textbf{ }} & \multicolumn{8}{c}{\textbf{Germany (EPEX-DE)}} & \multicolumn{8}{c}{\textbf{Spain (OMIE)}} \\
\cmidrule(l{3pt}r{3pt}){2-9} \cmidrule(l{3pt}r{3pt}){10-17}
\textbf{App.} & 1H & 2H & 3H & 4H & 6H & 8H & 12H & 24H & 1H & 2H & 3H & 4H & 6H & 8H & 12H & 24H\\
\midrule
\addlinespace[0.3em]
\multicolumn{17}{l}{\textbf{ARX forecasts}}\\
$base$ & 26.93 & 26.44 & 26.11 & 25.65 & 25.14 & 24.37 & 22.81 & 20.92 & 16.72 & 16.41 & 16.16 & 15.88 & 15.60 & 15.14 & 14.27 & 13.57\\
$wls$ & 26.92 & 26.42 & 26.07 & 25.59 & 25.05 & 24.33 & 22.79 & 21.00 & 16.72 & 16.40 & 16.15 & 15.86 & 15.54 & 15.07 & 14.26 & 13.42\\
$shr$ & 26.08 & 25.60 & 25.26 & 24.76 & 24.28 & 23.59 & 22.10 & 20.52 & 16.52 & 16.20 & 15.96 & 15.67 & 15.39 & 14.95 & 14.19 & 13.36\\
\addlinespace[0.3em]
\multicolumn{17}{l}{\textbf{NARX forecasts}}\\
$base$ & 23.06 & 22.64 & 22.12 & 21.75 & 21.08 & 20.54 & 19.37 & 17.99 & 16.57 & 16.29 & 16.01 & 15.68 & 15.27 & 15.00 & 14.25 & 13.29\\
$wls$ & 22.99 & 22.29 & 21.88 & 21.37 & 20.79 & 20.16 & 18.66 & 17.37 & 16.49 & 16.06 & 15.79 & 15.49 & 15.13 & 14.77 & 14.06 & 13.19\\
$shr$ & 22.27 & 21.59 & 21.14 & 20.60 & 19.94 & 19.39 & 17.72 & 16.41 & 16.34 & 15.89 & 15.63 & 15.36 & 14.95 & 14.61 & 13.87 & 13.00\\
\addlinespace[0.3em]
\multicolumn{17}{l}{\textbf{XGB forecasts}}\\
$base$ & 23.33 & 22.74 & 22.26 & 22.01 & 21.52 & 21.06 & 19.70 & 18.80 & 16.54 & 16.19 & 15.93 & 15.78 & 15.41 & 15.10 & 14.29 & 13.25\\
$wls$ & 23.25 & 22.63 & 22.21 & 21.74 & 21.18 & 20.62 & 19.08 & 17.77 & 16.51 & 16.10 & 15.80 & 15.50 & 15.16 & 14.74 & 14.05 & 13.21\\
$shr$ & 22.36 & 21.72 & 21.30 & 20.78 & 20.20 & 19.68 & 17.98 & 16.63 & 16.00 & 15.59 & 15.29 & 14.97 & 14.61 & 14.20 & 13.50 & 12.57\\
\addlinespace[0.3em]
\multicolumn{17}{l}{\textbf{MITRA forecasts}}\\
$base$ & 21.37 & 20.92 & 20.61 & 20.13 & 19.51 & 19.45 & 17.59 & 16.56 & 15.92 & 15.56 & 15.31 & 15.01 & 14.77 & 14.44 & 13.67 & 12.91\\
$wls$ & 21.42 & 20.64 & 20.19 & 19.64 & 18.99 & 18.47 & 16.74 & 15.36 & 15.88 & 15.39 & 15.09 & 14.80 & 14.39 & 13.97 & 13.24 & 12.36\\
$shr$ & \em{21.08} & \em{20.34} & \em{19.86} & \em{19.34} & \em{18.71} & \em{18.18} & \textbf{16.50} & \textbf{15.15} & 15.88 & 15.43 & 15.14 & 14.84 & 14.44 & 14.06 & 13.34 & 12.43\\
\addlinespace[0.3em]
\multicolumn{17}{l}{\textbf{Temporal coherent combination}}\\
$wls$ & 21.28 & 20.72 & 20.32 & 19.86 & 19.32 & 18.81 & 17.28 & 16.04 & \em{15.36} & \em{15.01} & \em{14.74} & \em{14.46} & \em{14.13} & \em{13.76} & \em{13.09} & \em{12.28}\\
$shr$ & 22.14 & 21.53 & 21.11 & 20.60 & 20.00 & 19.56 & 17.82 & 16.56 & 16.15 & 15.77 & 15.46 & 15.14 & 14.75 & 14.32 & 13.61 & 12.70\\
$shr_{bn}$ & 22.30 & 21.56 & 21.09 & 20.54 & 19.96 & 19.39 & 17.67 & 16.42 & 16.76 & 16.26 & 15.91 & 15.61 & 15.25 & 14.77 & 14.04 & 13.17\\
$shr_{be}$ & \textbf{20.49} & \textbf{19.94} & \textbf{19.54} & \textbf{19.06} & \textbf{18.56} & \textbf{18.08} & \em{16.56} & \em{15.28} & \textbf{15.04} & \textbf{14.68} & \textbf{14.42} & \textbf{14.13} & \textbf{13.78} & \textbf{13.41} & \textbf{12.77} & \textbf{11.92}\\
\bottomrule
\end{tabular}

%% file: tables/tab_shr_mae.tex
\begin{tabular}[t]{>{}c|ccccccc>{}c|cccccccc}
\toprule
\multicolumn{1}{c}{\textbf{ }} & \multicolumn{8}{c}{\textbf{Germany (EPEX-DE)}} & \multicolumn{8}{c}{\textbf{Spain (OMIE)}} \\
\cmidrule(l{3pt}r{3pt}){2-9} \cmidrule(l{3pt}r{3pt}){10-17}
\textbf{App.} & 1H & 2H & 3H & 4H & 6H & 8H & 12H & 24H & 1H & 2H & 3H & 4H & 6H & 8H & 12H & 24H\\
\midrule
\addlinespace[0.3em]
\multicolumn{17}{l}{\textbf{Linear Shrinkage \citep{ledoit2004, Schafer2005-jh}}}\\
$shr$ & 22.14 & 21.53 & 21.11 & 20.60 & 20.00 & 19.56 & 17.82 & 16.56 & 16.15 & 15.77 & 15.46 & 15.14 & 14.75 & 14.32 & 13.61 & 12.70\\
$shr_{bn}$ & 22.30 & 21.56 & 21.09 & 20.54 & 19.96 & 19.39 & 17.67 & 16.42 & 16.76 & 16.26 & 15.91 & 15.61 & 15.25 & 14.77 & 14.04 & 13.17\\
$shr_{be}$ & \textbf{20.49} & \textbf{19.94} & \textbf{19.54} & \textbf{19.06} & \textbf{18.56} & \textbf{18.08} & \textbf{16.56} & \textbf{15.28} & \textbf{15.04} & \textbf{14.68} & \textbf{14.42} & \textbf{14.13} & \textbf{13.78} & \textbf{13.41} & \textbf{12.77} & \textbf{11.92}\\
\addlinespace[0.3em]
\multicolumn{17}{l}{\textbf{Linear-Inverse Shrinkage \citep{ledoit2022power}}}\\
$lis$ & 22.84 & 22.15 & 21.65 & 21.08 & 20.47 & 19.95 & 18.23 & 16.76 & 16.69 & 16.25 & 15.93 & 15.57 & 15.15 & 14.71 & 13.92 & 13.02\\
$lis_{bn}$ & 22.44 & 21.69 & 21.22 & 20.66 & 20.07 & 19.49 & 17.76 & 16.51 & 16.89 & 16.36 & 16.02 & 15.71 & 15.34 & 14.85 & 14.11 & 13.24\\
$lis_{be}$ & \em{20.54} & \em{19.98} & \em{19.60} & \em{19.11} & \em{18.61} & \em{18.15} & \em{16.60} & \em{15.35} & \em{15.12} & \em{14.76} & \em{14.50} & \em{14.21} & \em{13.86} & \em{13.49} & \em{12.85} & \em{12.00}\\
\addlinespace[0.3em]
\multicolumn{17}{l}{\textbf{Quadratic-Inverse Shrinkage \citep{ledoit2022quadratic}}}\\
$qis$ & 22.81 & 22.12 & 21.63 & 21.05 & 20.45 & 19.93 & 18.21 & 16.75 & 16.67 & 16.24 & 15.91 & 15.56 & 15.14 & 14.70 & 13.91 & 13.01\\
$qis_{bn}$ & 22.44 & 21.69 & 21.22 & 20.66 & 20.07 & 19.49 & 17.76 & 16.51 & 16.89 & 16.36 & 16.02 & 15.71 & 15.34 & 14.85 & 14.11 & 13.24\\
$qis_{be}$ & 20.54 & 19.98 & 19.60 & 19.11 & 18.62 & 18.15 & 16.60 & 15.35 & 15.12 & 14.76 & 14.50 & 14.21 & 13.86 & 13.49 & 12.85 & 12.00\\
\addlinespace[0.3em]
\multicolumn{17}{l}{\textbf{Geometric-Inverse Shrinkage \citep{ledoit2022power}}}\\
$gis$ & 22.82 & 22.14 & 21.64 & 21.06 & 20.46 & 19.94 & 18.22 & 16.75 & 16.68 & 16.24 & 15.92 & 15.57 & 15.14 & 14.71 & 13.91 & 13.01\\
$gis_{bn}$ & 22.44 & 21.69 & 21.22 & 20.66 & 20.07 & 19.49 & 17.76 & 16.51 & 16.89 & 16.36 & 16.02 & 15.71 & 15.34 & 14.85 & 14.11 & 13.24\\
$gis_{be}$ & 20.54 & 19.98 & 19.60 & 19.11 & 18.62 & 18.15 & 16.60 & 15.35 & 15.12 & 14.76 & 14.50 & 14.21 & 13.86 & 13.49 & 12.85 & 12.00\\
\bottomrule
\end{tabular}

%% file: tables/tab_split_shrbe.tex
\begin{tabular}[t]{>{}c|ccccccc>{}c|cccccccc}
\toprule
\multicolumn{1}{c}{\textbf{ }} & \multicolumn{8}{c}{\textbf{Germany (EPEX-DE)}} & \multicolumn{8}{c}{\textbf{Spain (OMIE)}} \\
\cmidrule(l{3pt}r{3pt}){2-9} \cmidrule(l{3pt}r{3pt}){10-17}
\textbf{Excl. bf} & 1H & 2H & 3H & 4H & 6H & 8H & 12H & 24H & 1H & 2H & 3H & 4H & 6H & 8H & 12H & 24H\\
\midrule
\addlinespace[0.3em]
\multicolumn{17}{l}{\textbf{Mean Absolute Error}}\\
ARX & 20.57 & 19.98 & 19.59 & 19.10 & 18.59 & 18.11 & 16.56 & 15.29 & 15.14 & 14.77 & \em{14.50} & 14.20 & \em{13.85} & 13.47 & 12.84 & 11.98\\
NARX & 20.81 & 20.24 & 19.85 & 19.35 & 18.85 & 18.39 & 16.82 & 15.54 & 15.20 & 14.84 & 14.57 & 14.27 & 13.94 & 13.55 & 12.89 & 12.05\\
XGB & \textbf{20.36} & \textbf{19.78} & \textbf{19.37} & \textbf{18.88} & \textbf{18.33} & \textbf{17.81} & \textbf{16.24} & \textbf{14.94} & \em{15.11} & \em{14.74} & 14.50 & \em{14.20} & 13.86 & \em{13.47} & \em{12.80} & \em{11.96}\\
MITRA & 21.13 & 20.59 & 20.20 & 19.70 & 19.20 & 18.70 & 17.18 & 15.88 & 15.28 & 14.93 & 14.66 & 14.37 & 14.02 & 13.65 & 13.00 & 12.14\\
None & \em{20.49} & \em{19.94} & \em{19.54} & \em{19.06} & \em{18.56} & \em{18.08} & \em{16.56} & \em{15.28} & \textbf{15.04} & \textbf{14.68} & \textbf{14.42} & \textbf{14.13} & \textbf{13.78} & \textbf{13.41} & \textbf{12.77} & \textbf{11.92}\\
\addlinespace[0.3em]
\multicolumn{17}{l}{\textbf{Root Mean Square Error}}\\
ARX & 32.29 & 31.21 & 30.43 & 29.61 & 28.42 & 27.73 & 25.09 & 23.17 & 22.48 & 21.97 & 21.57 & 21.20 & 20.69 & 20.08 & 19.03 & 17.81\\
NARX & 32.63 & 31.58 & 30.82 & 29.98 & 28.84 & 28.12 & 25.51 & 23.57 & 22.62 & 22.12 & 21.72 & 21.36 & 20.85 & 20.24 & 19.19 & 17.94\\
XGB & \textbf{31.72} & \textbf{30.64} & \textbf{29.84} & \textbf{28.98} & \textbf{27.76} & \textbf{27.02} & \textbf{24.35} & \textbf{22.29} & \textbf{22.31} & \textbf{21.80} & \textbf{21.40} & \textbf{21.04} & \textbf{20.52} & \textbf{19.92} & \textbf{18.84} & \textbf{17.63}\\
MITRA & 33.15 & 32.12 & 31.38 & 30.54 & 29.40 & 28.66 & 26.11 & 24.20 & 22.66 & 22.17 & 21.77 & 21.40 & 20.90 & 20.29 & 19.22 & 17.96\\
None & \em{32.14} & \em{31.10} & \em{30.33} & \em{29.51} & \em{28.34} & \em{27.64} & \em{25.04} & \em{23.12} & \em{22.36} & \em{21.86} & \em{21.47} & \em{21.11} & \em{20.60} & \em{20.01} & \em{18.95} & \em{17.73}\\
\bottomrule
\end{tabular}

%% file: tables/tab_seq_mae.tex
\begin{tabular}[t]{>{}c|ccccccc>{}c|cccccccc}
\toprule
\multicolumn{1}{c}{\textbf{ }} & \multicolumn{8}{c}{\textbf{Germany (EPEX-DE)}} & \multicolumn{8}{c}{\textbf{Spain (OMIE)}} \\
\cmidrule(l{3pt}r{3pt}){2-9} \cmidrule(l{3pt}r{3pt}){10-17}
\textbf{App.} & 1H & 2H & 3H & 4H & 6H & 8H & 12H & 24H & 1H & 2H & 3H & 4H & 6H & 8H & 12H & 24H\\
\midrule
\addlinespace[0.3em]
\multicolumn{17}{l}{\textbf{Forecast combination and sequential approaches}}\\
$sa^\ast$ & 21.87 & 21.39 & 20.93 & 20.55 & 20.04 & 19.61 & 18.21 & 16.95 & 15.41 & 15.08 & 14.83 & 14.55 & 14.23 & 13.90 & 13.14 & 12.27\\
$sa+wls$ & 21.85 & 21.33 & 20.92 & 20.48 & 19.95 & 19.42 & 17.96 & 16.62 & 15.39 & 15.04 & 14.79 & 14.51 & 14.18 & 13.81 & 13.11 & 12.31\\
$sa+shr$ & 21.13 & 20.60 & 20.21 & 19.72 & 19.27 & 18.73 & 17.26 & 15.91 & \em{15.13} & \em{14.78} & \em{14.54} & \em{14.25} & \em{13.91} & \em{13.56} & \em{12.87} & \em{12.02}\\
\addlinespace[0.3em]
\multicolumn{17}{l}{\textbf{Best reconciliation approach (MITRA base forecasts)}}\\
$shr$ & \em{21.08} & \em{20.34} & \em{19.86} & \em{19.34} & \em{18.71} & \em{18.18} & \textbf{16.50} & \textbf{15.15} & 15.88 & 15.43 & 15.14 & 14.84 & 14.44 & 14.06 & 13.34 & 12.43\\
\addlinespace[0.3em]
\multicolumn{17}{l}{\textbf{Best temporal coherent combination approach}}\\
$shr_{be}$ & \textbf{20.49} & \textbf{19.94} & \textbf{19.54} & \textbf{19.06} & \textbf{18.56} & \textbf{18.08} & \em{16.56} & \em{15.28} & \textbf{15.04} & \textbf{14.68} & \textbf{14.42} & \textbf{14.13} & \textbf{13.78} & \textbf{13.41} & \textbf{12.77} & \textbf{11.92}\\
\bottomrule
\end{tabular}

%% file: tables/tab_covariance_rmse.tex
\begin{tabular}[t]{>{}c|ccccccc>{}c|cccccccc}
\toprule
\multicolumn{1}{c}{\textbf{ }} & \multicolumn{8}{c}{\textbf{Germany (EPEX-DE)}} & \multicolumn{8}{c}{\textbf{Spain (OMIE)}} \\
\cmidrule(l{3pt}r{3pt}){2-9} \cmidrule(l{3pt}r{3pt}){10-17}
\textbf{App.} & 1H & 2H & 3H & 4H & 6H & 8H & 12H & 24H & 1H & 2H & 3H & 4H & 6H & 8H & 12H & 24H\\
\midrule
\addlinespace[0.3em]
\multicolumn{17}{l}{\textbf{ARX forecasts}}\\
$base$ & 39.83 & 38.95 & 38.28 & 37.50 & 36.42 & 35.59 & 32.68 & 30.28 & 24.40 & 23.96 & 23.59 & 23.25 & 22.77 & 22.12 & 20.83 & 19.50\\
$wls$ & 39.82 & 38.93 & 38.24 & 37.41 & 36.30 & 35.44 & 32.72 & 30.12 & 24.43 & 23.96 & 23.57 & 23.21 & 22.67 & 22.01 & 20.75 & 19.34\\
$shr$ & 38.56 & 37.66 & 36.98 & 36.13 & 35.12 & 34.15 & 31.64 & 29.15 & 24.28 & 23.82 & 23.43 & 23.07 & 22.55 & 21.90 & 20.68 & 19.31\\
\addlinespace[0.3em]
\multicolumn{17}{l}{\textbf{NARX forecasts}}\\
$base$ & 35.62 & 34.84 & 33.89 & 33.27 & 32.00 & 31.20 & 28.95 & 27.24 & 24.07 & 23.71 & 23.44 & 22.96 & 22.41 & 21.92 & 20.64 & 19.14\\
$wls$ & 35.55 & 34.36 & 33.55 & 32.68 & 31.45 & 30.66 & 28.03 & 26.13 & 24.03 & 23.39 & 22.96 & 22.58 & 22.03 & 21.45 & 20.35 & 19.10\\
$shr$ & 34.56 & 33.37 & 32.55 & 31.67 & 30.42 & 29.61 & 26.94 & 24.90 & 24.02 & 23.39 & 22.96 & 22.57 & 22.02 & 21.42 & 20.28 & 19.01\\
\addlinespace[0.3em]
\multicolumn{17}{l}{\textbf{XGB forecasts}}\\
$base$ & 37.05 & 35.87 & 35.10 & 34.24 & 33.38 & 32.94 & 30.55 & 29.50 & 24.57 & 24.11 & 23.75 & 23.60 & 22.93 & 22.50 & 21.24 & 19.61\\
$wls$ & 36.82 & 35.73 & 34.97 & 34.13 & 33.00 & 32.25 & 29.80 & 27.93 & 24.53 & 23.96 & 23.56 & 23.17 & 22.66 & 22.04 & 20.99 & 19.76\\
$shr$ & 35.13 & 33.99 & 33.20 & 32.29 & 31.15 & 30.34 & 27.61 & 25.60 & 23.70 & 23.13 & 22.71 & 22.29 & 21.77 & 21.12 & 20.01 & 18.66\\
\addlinespace[0.3em]
\multicolumn{17}{l}{\textbf{MITRA forecasts}}\\
$base$ & 33.31 & 32.43 & 31.87 & 30.86 & 29.73 & 29.34 & 26.27 & 24.81 & 23.61 & 23.17 & 22.72 & 22.19 & 21.87 & 21.45 & 20.16 & 18.84\\
$wls$ & 33.38 & 32.01 & 31.09 & 30.16 & 28.84 & 28.07 & \em{25.04} & \em{22.92} & 23.54 & 22.84 & 22.38 & 21.99 & 21.39 & 20.74 & 19.63 & 18.34\\
$shr$ & \em{32.84} & \em{31.52} & \em{30.61} & \em{29.70} & \em{28.38} & \textbf{27.63} & \textbf{24.69} & \textbf{22.56} & 23.25 & 22.59 & 22.14 & 21.75 & 21.18 & 20.53 & 19.44 & 18.17\\
\addlinespace[0.3em]
\multicolumn{17}{l}{\textbf{Temporal coherent combination}}\\
$wls$ & 33.09 & 32.03 & 31.27 & 30.45 & 29.27 & 28.59 & 25.98 & 24.09 & \em{22.74} & \em{22.24} & \em{21.85} & \em{21.49} & \em{20.98} & \em{20.39} & \em{19.34} & \em{18.11}\\
$shr$ & 34.82 & 33.73 & 32.92 & 32.06 & 30.83 & 30.12 & 27.38 & 25.39 & 23.80 & 23.26 & 22.83 & 22.41 & 21.86 & 21.19 & 20.03 & 18.73\\
$shr_{bn}$ & 34.74 & 33.42 & 32.57 & 31.65 & 30.38 & 29.63 & 26.76 & 24.82 & 24.71 & 24.01 & 23.56 & 23.14 & 22.58 & 21.91 & 20.82 & 19.57\\
$shr_{be}$ & \textbf{32.14} & \textbf{31.10} & \textbf{30.33} & \textbf{29.51} & \textbf{28.34} & \em{27.64} & 25.04 & 23.12 & \textbf{22.36} & \textbf{21.86} & \textbf{21.47} & \textbf{21.11} & \textbf{20.60} & \textbf{20.01} & \textbf{18.95} & \textbf{17.73}\\
\bottomrule
\end{tabular}

%% file: tables/tab_shr_rmse.tex
\begin{tabular}[t]{>{}c|ccccccc>{}c|cccccccc}
\toprule
\multicolumn{1}{c}{\textbf{ }} & \multicolumn{8}{c}{\textbf{Germany (EPEX-DE)}} & \multicolumn{8}{c}{\textbf{Spain (OMIE)}} \\
\cmidrule(l{3pt}r{3pt}){2-9} \cmidrule(l{3pt}r{3pt}){10-17}
\textbf{App.} & 1H & 2H & 3H & 4H & 6H & 8H & 12H & 24H & 1H & 2H & 3H & 4H & 6H & 8H & 12H & 24H\\
\midrule
\addlinespace[0.3em]
\multicolumn{17}{l}{\textbf{Linear Shrinkage \citep{ledoit2004, Schafer2005-jh}}}\\
$shr$ & 34.82 & 33.73 & 32.92 & 32.06 & 30.83 & 30.12 & 27.38 & 25.39 & 23.80 & 23.26 & 22.83 & 22.41 & 21.86 & 21.19 & 20.03 & 18.73\\
$shr_{bn}$ & 34.74 & 33.42 & 32.57 & 31.65 & 30.38 & 29.63 & 26.76 & 24.82 & 24.71 & 24.01 & 23.56 & 23.14 & 22.58 & 21.91 & 20.82 & 19.57\\
$shr_{be}$ & \textbf{32.14} & \textbf{31.10} & \textbf{30.33} & \textbf{29.51} & \textbf{28.34} & \textbf{27.64} & \textbf{25.04} & \textbf{23.12} & \textbf{22.36} & \textbf{21.86} & \textbf{21.47} & \textbf{21.11} & \textbf{20.60} & \textbf{20.01} & \textbf{18.95} & \textbf{17.73}\\
\addlinespace[0.3em]
\multicolumn{17}{l}{\textbf{Linear-Inverse Shrinkage \citep{ledoit2022power}}}\\
$lis$ & 35.90 & 34.71 & 33.84 & 32.92 & 31.61 & 30.85 & 28.02 & 25.88 & 24.64 & 24.02 & 23.55 & 23.06 & 22.51 & 21.74 & 20.49 & 19.13\\
$lis_{bn}$ & 34.95 & 33.61 & 32.76 & 31.82 & 30.55 & 29.79 & 26.91 & 24.96 & 24.98 & 24.25 & 23.78 & 23.36 & 22.80 & 22.11 & 21.01 & 19.76\\
$lis_{be}$ & \em{32.26} & \em{31.22} & \em{30.45} & \em{29.63} & \em{28.46} & \em{27.74} & \em{25.14} & \em{23.23} & \em{22.49} & \em{22.00} & \em{21.60} & \em{21.23} & \em{20.74} & \em{20.12} & \em{19.07} & \em{17.85}\\
\addlinespace[0.3em]
\multicolumn{17}{l}{\textbf{Quadratic-Inverse Shrinkage \citep{ledoit2022quadratic}}}\\
$qis$ & 35.85 & 34.67 & 33.80 & 32.89 & 31.58 & 30.82 & 28.00 & 25.86 & 24.62 & 24.00 & 23.53 & 23.04 & 22.48 & 21.73 & 20.47 & 19.12\\
$qis_{bn}$ & 34.95 & 33.61 & 32.76 & 31.82 & 30.55 & 29.79 & 26.91 & 24.96 & 24.98 & 24.25 & 23.78 & 23.36 & 22.80 & 22.11 & 21.01 & 19.76\\
$qis_{be}$ & 32.26 & 31.22 & 30.45 & 29.63 & 28.46 & 27.75 & 25.14 & 23.23 & 22.50 & 22.00 & 21.61 & 21.23 & 20.74 & 20.12 & 19.07 & 17.86\\
\addlinespace[0.3em]
\multicolumn{17}{l}{\textbf{Geometric-Inverse Shrinkage \citep{ledoit2022power}}}\\
$gis$ & 35.88 & 34.69 & 33.82 & 32.90 & 31.59 & 30.83 & 28.01 & 25.87 & 24.63 & 24.01 & 23.54 & 23.05 & 22.49 & 21.73 & 20.48 & 19.13\\
$gis_{bn}$ & 34.95 & 33.61 & 32.76 & 31.82 & 30.55 & 29.79 & 26.91 & 24.96 & 24.98 & 24.25 & 23.78 & 23.36 & 22.80 & 22.11 & 21.01 & 19.76\\
$gis_{be}$ & 32.26 & 31.22 & 30.45 & 29.63 & 28.46 & 27.74 & 25.14 & 23.23 & 22.49 & 22.00 & 21.60 & 21.23 & 20.74 & 20.12 & 19.07 & 17.86\\
\bottomrule
\end{tabular}

%% file: tables/tab_seq_rmse.tex
\begin{tabular}[t]{>{}c|ccccccc>{}c|cccccccc}
\toprule
\multicolumn{1}{c}{\textbf{ }} & \multicolumn{8}{c}{\textbf{Germany (EPEX-DE)}} & \multicolumn{8}{c}{\textbf{Spain (OMIE)}} \\
\cmidrule(l{3pt}r{3pt}){2-9} \cmidrule(l{3pt}r{3pt}){10-17}
\textbf{App.} & 1H & 2H & 3H & 4H & 6H & 8H & 12H & 24H & 1H & 2H & 3H & 4H & 6H & 8H & 12H & 24H\\
\midrule
\addlinespace[0.3em]
\multicolumn{17}{l}{\textbf{Forecast combination and sequential approaches}}\\
$sa^\ast$ & 33.65 & 32.68 & 31.94 & 31.13 & 30.00 & 29.42 & 26.85 & 25.20 & 22.73 & 22.28 & 21.92 & 21.53 & 21.05 & 20.51 & 19.36 & 17.94\\
$sa+wls$ & 33.60 & 32.60 & 31.86 & 31.06 & 29.91 & 29.22 & 26.67 & 24.69 & 22.71 & 22.22 & 21.83 & 21.48 & 20.97 & 20.39 & 19.31 & 18.06\\
$sa+shr$ & \em{32.55} & 31.52 & 30.76 & 29.94 & 28.77 & 28.07 & 25.49 & 23.48 & \textbf{22.25} & \textbf{21.75} & \textbf{21.35} & \textbf{21.00} & \textbf{20.48} & \textbf{19.88} & \textbf{18.79} & \textbf{17.51}\\
\addlinespace[0.3em]
\multicolumn{17}{l}{\textbf{Best reconciliation approach (MITRA base forecasts)}}\\
$shr$ & 32.84 & \em{31.52} & \em{30.61} & \em{29.70} & \em{28.38} & \textbf{27.63} & \textbf{24.69} & \textbf{22.56} & 23.25 & 22.59 & 22.14 & 21.75 & 21.18 & 20.53 & 19.44 & 18.17\\
\addlinespace[0.3em]
\multicolumn{17}{l}{\textbf{Best temporal coherent combination approach}}\\
$shr_{be}$ & \textbf{32.14} & \textbf{31.10} & \textbf{30.33} & \textbf{29.51} & \textbf{28.34} & \em{27.64} & \em{25.04} & \em{23.12} & \em{22.36} & \em{21.86} & \em{21.47} & \em{21.11} & \em{20.60} & \em{20.01} & \em{18.95} & \em{17.73}\\
\bottomrule
\end{tabular}